\documentclass[10pt]{amsart}
\usepackage{amsmath, amsfonts, amsthm, amssymb}
\usepackage[a4paper, width=16cm, top=30mm, bottom=30mm]{geometry}
\usepackage{epstopdf}
\usepackage{graphicx}
\usepackage{url}
\usepackage{neuralnetwork}
\usepackage{verbatim}
\usepackage{hyperref}
\usepackage{color}
\usepackage{bbm}
\usepackage{ulem}
\usepackage{multirow}
\usepackage{enumerate}
\usepackage[]{mcode}
\usepackage{braket}
\usepackage{tikz}
\usepackage{subfig}
\usepackage{algorithm}
\usepackage{algpseudocode}
\graphicspath{{Figures/}}

\usepackage{savesym}
\savesymbol{link}
\usepackage{qcircuit}
\restoresymbol{link}{link}

\newtheorem{theorem}{Theorem}[section]

\newtheorem{proposition}[theorem]{Proposition}

\theoremstyle{definition}
\newtheorem{definition}[theorem]{Definition}
\newtheorem{remark}[theorem]{Remark}

\newtheorem{example}[theorem]{Example}
\numberwithin{equation}{section}

\newcommand{\ind}{1\hspace{-2.1mm}{1}}
\newcommand{\CC}{\mathbb{C}}
\newcommand{\RR}{\mathbb{R}}
\newcommand{\PP}{\mathbb{P}}
\newcommand{\QQ}{\mathbb{Q}}
\newcommand{\D}{\mathrm{d}}
\newcommand{\I}{\mathrm{i}}
\newcommand{\Cr}{\mathrm{C}}
\newcommand{\Put}{\mathrm{P}}
\newcommand{\SVI}{\mathrm{SVI}}

\newcommand{\Mm}{\mathcal{M}}
\newcommand{\Oo}{\mathcal{O}}

\newcommand{\Xx}{\mathcal{X}}

\newcommand{\QPE}{\mathrm{QPE}}
\newcommand{\BS}{\mathrm{BS}}

\newcommand{\QFT}{\mathrm{QFT}}

\newcommand{\KL}{\mathrm{KL}}
\newcommand{\NN}{\mathbb{N}}

\newcommand{\Ff}{\mathcal{F}}

\newcommand{\Hh}{\mathcal{H}}
\newcommand{\EE}{\mathbb{E}}
\newcommand{\Dd}{\mathcal{D}}
\newcommand{\ek}{\mathfrak{e}}
\newcommand{\Dr}{\mathfrak{D}}
\newcommand{\Ss}{\mathcal{S}}

\newcommand{\Rr}{\mathcal{R}}
\newcommand{\ww}{\boldsymbol{\mathrm{w}}}
\newcommand{\xx}{\boldsymbol{\mathrm{x}}}
\newcommand{\yy}{\boldsymbol{\mathrm{y}}}
\newcommand{\XX}{\boldsymbol{\mathrm{X}}}
\newcommand{\xxw}{\widetilde{\xx}}

\newcommand{\thbg}{\boldsymbol{\theta}_{G}}
\newcommand{\thbd}{\boldsymbol{\theta}_{D}}
\newcommand{\thbgs}{\boldsymbol{\theta}_{G}^*}

\newcommand{\wwd}{\ww_{D}}

\newcommand{\half}{\frac{1}{2}}
\newcommand{\Ffq}{{}_{\mathrm{q}}\Ff}
\newcommand{\qt}{m}

\newcommand{\eps}{\varepsilon}

\newcommand{\tar}{\mathrm{target}}
\newcommand{\imp}{\mathrm{imp}}

\newcommand{\out}{\mathrm{out}}
\newcommand{\fake}{\textit{Fake}}
\newcommand{\real}{\textit{Real}}
\newcommand{\Df}{\textbf{D}}
\newcommand{\Gf}{\textbf{G}}
\newcommand{\Id}{\mathrm{I_{d}}}

\newcommand{\E}{\mathrm{e}}
\newcommand{\Tr}{\mathrm{Tr}}

\newcommand{\CR}{{}_{\mathrm{c}}\mathrm{R}}
\newcommand{\CRY}{\CR_{\mathrm{Y}}}
\newcommand{\Hg}{\mathrm{H}}
\newcommand{\Cg}{\mathrm{C}}
\newcommand{\Ug}{\mathrm{U}}
\newcommand{\Yg}{\mathrm{Y}}
\newcommand{\Xg}{\mathrm{X}}
\newcommand{\RY}{\mathrm{R}_{\mathrm{Y}}}

\begin{document}

\title{A Quantum Generative Adversarial Network for Distributions}

\author{Amine Assouel}
\address{ENS Paris-Saclay}
\email{amine.assouel@ens-paris-saclay.fr}

\author{Antoine Jacquier}
\address{Department of Mathematics, Imperial College London, and Alan Turing Institute}
\email{a.jacquier@imperial.ac.uk}

\author{Alexei Kondratyev}
\address{Abu Dhabi Investment Authority (ADIA), and Department of Mathematics, Imperial College London}
\email{a.kondratyev@imperial.ac.uk}

\date{\today}
\thanks{The authors would like to thank Konstantinos Kardaras and Alexandros Pavlis for insightful discussion on quantum algorithms and spins.
AJ acknowledges financial support from the EPSRC EP/T032146/1 grant.}
\keywords{Quantum Computing, GAN, Quantum Phase Estimation, SVI, volatility}
\subjclass[2010]{81P68, 81-08, 91G20}

\maketitle
\begin{abstract}
Generative Adversarial Networks are becoming a fundamental tool in Machine Learning, in particular in the context of improving the stability of deep neural networks.
At the same time, recent advances in Quantum Computing have shown that, despite the absence of a fault-tolerant quantum computer so far, quantum techniques are providing exponential advantage over their classical counterparts.
We develop a fully connected Quantum Generative Adversarial network and show how it can be applied in Mathematical Finance, with a particular focus on volatility modelling.
\end{abstract}

\section{Introduction}
Machine Learning has become ubiquitous, with applications in nearly every aspects of  society today, in particular for image and speech recognition, traffic prediction, product recommendation, medical diagnosis, stock market trading, fraud detection.
One specific Machine Learning tool, deep neural networks, has seen tremendous  developments over the past few year.
Despite clear advances, these networks however often suffer from the lack of training data: in Finance, time series of a stock price only occur once, physical experiments are sometimes expensive to run many times...
To palliate this, attention has turn to methods aimed at reproducing existing data with a high degree of accuracy. 
Among these, Generative Adversarial Networks (GAN) are a class of unsupervised Machine Learning devices whereby two neural networks, a generator and a discriminator, contest against each other in a minimax game in order to generate information similar to a given dataset~\cite{goodfellow2014generative}. 
They have been successfully applied in many fields over the past few years, in particular for image generation~\cite{GenImages, GenImages2}, medicine~\cite{Med1, Med2}, and in Quantitative Finance~\cite{Ruf2020}. 
They however often suffer from instability issues, vanishing gradient and potential mode collapse~\cite{modecollapse}. 
Even Wasserstein GANs, assuming the Wasserstein distance from optimal transport instead of the classical Jensen–Shannon Divergence, are still subject to slow convergence issues and potential instability~\cite{GANImprov3}.

In order to improve the accuracy of this method, Lloyd and Weedbrook~\cite{QuGAN} 
and Dallaire-Demers and Killoran~\cite{QuGANDallaire}
simultaneously introduced a quantum component to GANs, 
where the data consists of quantum states or classical data while the two players are equipped with quantum information processors. 
Preliminary works have demonstrated the quality of this approach, 
in particular for high-dimensional data, 
thus leveraging on the exponential advantage of quantum computing~\cite{advQuGAN}. 
An experimental proof-of-principle  demonstration of QuGAN in
a superconducting quantum circuit was shown in~\cite{QuGANSuper},
while in~\cite{QuGAN2} the authors made use 
of quantum fidelity measurements to propose a loss function acting on quantum states.
Further recent advances, providing more insights on how quantum entanglement can play a decisive role, have been put forward in~\cite{QENtang}.
While actual Quantum computers are not available yet, Noisy intermediate-scale quantum (NISQ) algorithms are already here and allow us to perform quantum-like operations\cite{NISQ}. 

We focus here on building a fully connected Quantum Generative Adversarial network (QuGAN)~\footnote{The terminology `QuGAN' should not be confused with `QGAN', used to denote quantised versions of GAN, as in~\cite{QuantizedGAN}, nor with `Quant GAN', which refers~\cite{QuantGAN} to the use of GAN in Quantitative Finance; neither Quant GAN nor QGAN are related whatsoever to Quantum Computing.}, 
namely an entire quantum counterpart to a classical GAN.
A quantum version of GAN was first introduced in~\cite{QuGAN} and~\cite{QuGANDallaire}, 
showing that it may exhibit an exponential advantage
over classical adversarial networks.

The paper is structured as follows: 
In Section~\ref{sec:Network}, we recall the basics of a classical neural network and show how to build a fully quantum version of it.
This is incorporated in the full architecture of a Quantum Generative Adversarial Network in Section~\ref{sec:QuGAN}.
Since classical GANs are becoming an important focus in Quantitative Finance~\cite{FinGAN1, FinGAN2, FinGAN3, QuantGAN}, 
we provide an example of application for QuGAN for volatility modelling in Section~\ref{sec:FinAppli}, hoping to bridge the gap between the Quantum Computing and the Quantitative Finance communities.
For completeness, we gather a few useful results from Quantum Computing in Appendix~\ref{sec:Review}.

\section{A quantum version of a non-linear quantum neuron}\label{sec:Network}

The quantum phase estimation procedure lies at the very core of building a quantum counterpart for a neural network. In this part, we will mainly focus on how to build a single quantum neuron.
As the fundamental building block of artificial neural networks,
a neuron classically maps a normalised input 
$\xx=(x_0,\ldots,x_{n-1})^{\top} \in [0,1]^n$ to an output $g(\xx^{\top} \ww)$, 
where $\ww = (w_0,\ldots,w_{n-1})^{\top} \in [-1,1]^n$ is the weight vector, for some activation function~$g$.
The non-linear quantum neuron requires the following steps:
\begin{itemize}
\item Encode classical data into quantum states (Section~\ref{sec:QEncoding});
\item Perform the (quantum version of the) inner product $\xx^{\top}\ww$ (Section~\ref{sec:QInnerProd});
\item Applying the (quantum version of the) non-linear activation function (Section~\ref{sec:QActivation}).
\end{itemize}

Before diving into the quantum version of neural networks, 
we recall the basics of classical (feedforward) neural networks,
which we aim at mimicking.

\subsection{Classical neural network architecture}

Artificial neural networks (ANNs) are a subset of machine learning and lie at the heart of Deep Learning algorithms. Their name and structure are inspired by the human brain~\cite{NeuroML}, mimicking the way that biological neurons signal to one another.
They consist of several layers, with an input layer, one or more hidden layers, and an output layer, each one of them containing several nodes. 
An example of ANN is depicted in Figure~\ref{fig:neural_net}.
\begin{figure}[h!]
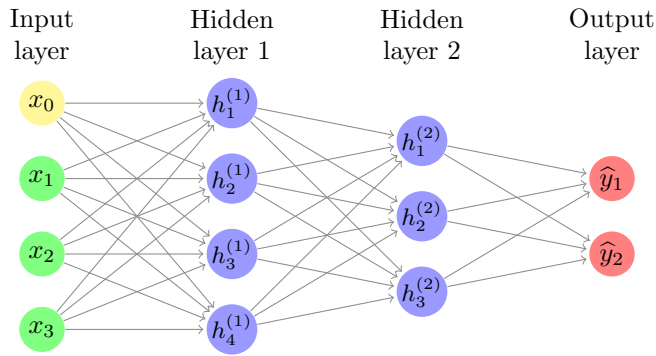

    \centering
    \begin{neuralnetwork}[height=4]
        \newcommand{\x}[2]{$x_#2$}
        \newcommand{\y}[2]{$\widehat{y}_#2$}
        \newcommand{\hfirst}[2]{\small $h^{(1)}_#2$}
        \newcommand{\hsecond}[2]{\small $h^{(2)}_#2$}
        \inputlayer[count=3, bias=true, title=Input\\layer, text=\x]
        \hiddenlayer[count=4, bias=false, title=Hidden\\layer 1, text=\hfirst] \linklayers
        \hiddenlayer[count=3, bias=false, title=Hidden\\layer 2, text=\hsecond] \linklayers
        \outputlayer[count=2, title=Output\\layer, text=\y] \linklayers
\end{neuralnetwork}
    \caption{ANN with one input layer, 2 hidden layers and one output layer.}
    \label{fig:neural_net}
\end{figure}

For a given an input vector  $\xx = (x_1,\ldots,x_n)\in\RR^n$, the connectivity between~$\xx$ 
and the $j$-th neuron~$h^{(1)}_j$ of the first hidden layer
(Figure~\ref{fig:neural_net}) is done via 
$h^{(1)}_j=\sigma_{1,j}(b_{1,j}+\sum_{i=1}^{n} x_{i}w_{i,j})$,
where $\sigma_{1,j}$ is called the activation function.
By denoting  $H_k\in\RR^{s_k}$ the vector of the k-th hidden layer, where $s_k\in\mathbb{N^*}$ and $H_k=(h^{(k)}_1,\ldots,h^{(k)}_{s_k})$ the connectivity model generalises itself to the whole network:
\begin{equation}\label{equation neural network}
   h_j^{(k+1)}=\sigma_{k+1,j}\left(b_{k+1,j}+\sum_{i=1}^{s_k} h_{i}^{(k)}w_{i,k+1,j}\right),
\end{equation}
where $j\in\{1,\ldots,s_{k+1}\}$.
Therefore for $l$ hidden layers the entire network is parameterised by 
$\Omega=(\sigma_{k,r_k},b_{k,r_k},w_{v_k,k,r_k})_{k,r_k,v_k}$ where first $1\leq k \leq l $, 
then $1\leq  r_k\leq s_k$ and $1\leq  v_k \leq s_{k-1}$. 
For a given training data set of size $N$, $(X_i,Y_i)_{i=1,\ldots,N}$, 
the goal of a neural network is to build a mapping between $(X_i)_{i=1,\ldots,N}$ and $(Y_i)_{i=1,\ldots,N}$.
The idea for the neural network structure comes from the Kolmogorov-Arnold representation Theorem~\cite{Arnold_1957, Kolmogorov_1956}:
\begin{theorem}\label{Kolmogorov Theorem}
Let $f: [0,1]^d\xrightarrow[]{}\RR$ be a continuous function. There exist sequences
$(\Phi_i)_{i=1, \ldots, 2d}$  and $(\Psi_{i,j})_{i=1, \ldots, 2d; i=1,\ldots, d}$ of continuous functions from~$\RR$ to~$\RR$ such that for all $(x_1,\ldots,x_d)\in[0,1]^d$,
\begin{equation}\label{Kolmogorov}
      f(x_1,\ldots,x_d)=\sum_{i=1}^{2d}\Phi_i\left(\sum_{j=1}^{d}\Psi_{i,j}(x_j)\right).
\end{equation}
\end{theorem}
The representation of~$f$ resembles a two-hidden-layer ANN, where $\Phi_i ,\Psi_{i,j}$ are the activation functions.

\subsection{Quantum encoding}\label{sec:QEncoding}
Since a quantum computer only takes qubits as inputs, 
we first need to encode the classical data into a quantum state.
For $x_{j}\in [0,1]$ and $p\in\NN$,
denote by $\frac{x_{j,1}}{2} + \frac{x_{j,2}}{2^2} + \ldots + \frac{x_{j,p}}{2^p}$ 
the $p$-binary approximation of~$x_{j}$,
where each $x_{j,k}$ belongs to $\{0,1\}$, for $k\in\{1,2,\ldots,p\}$. 
The quantum code for the classical value~$x_j$ is then defined via this approximation as
$$
\ket{x_j} := \ket{x_{j,1}}\otimes\ket{x_{j,2}}\otimes\ldots\otimes\ket{x_{j,p}}=\ket{x_{j,1}x_{j,2}\ldots x_{j,p}},
$$
and therefore the encoding for the vector $\xx$ is 
\begin{equation}\label{eq:xVecBin}
\ket{\xx}
 := \ket{x_{0,1} x_{0,2}\ldots x_{0,p}}\otimes\ldots\otimes\ket{x_{n-1,1}\ldots x_{n-1,p}}.
\end{equation}

\subsection{Quantum inner product}\label{sec:QInnerProd}
We now show how to build the quantum version of the inner product performing the operation
$$
\ket{0}^{\otimes m}\ket{\xx}\xrightarrow[]{}\ket{\xxw^{\top} \ww}\ket{\xx}.
$$

Denote the two-qubit controlled Z-Rotation gate by
$$
\CR_z(\alpha)= \begin{pmatrix}
1 & 0 & 0 & 0\\
0 & 1 & 0 & 0\\
0 & 0 & 1 & 0\\
0 & 0 & 0 & \E^{2\I\pi \alpha}
\end{pmatrix},
$$
where~$\alpha$ is the phase shift with period~$\pi$.
For $x\in\{0,1\}$ and $\ket{+}:=\frac{1}{\sqrt{2}}(\ket{0}+\ket{1})$, note that, for $k\in\NN$,
$$
\CR_z\left(\frac{1}{2^k}\right)
\left(\ket{+}\ket{x}\right)
=\frac{1}{\sqrt{2}}\left(\ket{0}\ket{x} + \exp\left\{\frac{2\I\pi x}{2^k}\right\}\ket{1}\ket{x}\right)
$$
Indeed, either $x=0$ and then $\ket{x}=\ket{0}$ so that
$$
\CR_z\left(\frac{1}{2^k}\right)
\left(\ket{+}\ket{x}\right)
=\frac{1}{\sqrt{2}}
\left(\ket{0}\ket{0}+\ket{1}\ket{0}\right),
$$
or $x=1$ and hence
$$
\CR_z\left(\frac{1}{2^k}\right)
\left(\ket{+}\ket{x}\right)
=\frac{1}{\sqrt{2}}\left(\ket{0}\ket{1} + \exp\left\{\frac{2\I\pi}{2^k}\right\}\ket{1}\ket{1}\right).
$$
The gate $\CR_z\left(\alpha\right)$ applies to two qubits where the first one constitutes what is called an ancilla qubit since it controls the computation. From there one should define the ancilla register that is composed of all the qubits that are used as controlled qubits.

\subsubsection{The case where with $m$ ancilla qubits and $\xx^{\top}\ww \in \{0,\ldots, 2^m-1\}$.}

The first part of the circuit consists of applying Hadamard gates on the 
ancilla register $\ket{0}^{\otimes m}$, which produces
\begin{equation}\label{output inner}
\Hg^{\otimes m}\ket{0}^{\otimes m}\ket{\xx}
=\left(\frac{1}{\sqrt{2^m}}\sum_{j=0}^{2^m-1}\ket{j}\right)\otimes\ket{\xx}.
\end{equation} 

The goal here is then to encode as a phase the result of the inner product $\xx^{\top} \ww$.
With the binary approximation~\eqref{eq:xVecBin} for ~$\ket{\xx}$ 
and~$m$ ancilla qubits, define for $l \in \{1,\ldots,m\}$, $j\in\{0,\ldots,n-1\}$ and $k \in \{1,\ldots,p\}$, 
$\CR_z^{l,j,k}\left(\alpha\right)$, the $\CR_z(\alpha)$ matrix applied to  the qubit $\ket{x_{j,k}}$ with the $l$-th qubit of the ancilla register as control.
Finally, introduce the unitary operator
\begin{equation}\label{eq:GatesUwm}
\Ug_{\ww,m}
 := \prod_{l=0}^{m-1}\left\{\prod_{j=0}^{n-1}\prod_{k=1}^{p}\CR_z^{m-l,j,k}\left(\frac{w_j}{2^{m+k}}\right)\right\}^{m-l}.
\end{equation}

\begin{proposition}
The following identity holds for all $n,p,m \in\NN$:
\begin{equation}\label{comparison}
\Ug_{\ww,m}\Hg^{\otimes m}\ket{0}^{\otimes m}\ket{\xx}
= \left(\frac{1}{\sqrt{2^m}}\sum_{j=0}^{2^m-1}
    \exp\left\{2\I\pi j \frac{\xxw^{\top}\ww}{2^m}\right\}\ket{j}\right)\otimes\ket{\xx},
\end{equation}
where
$$
\xxw^{\top}\ww
 := \sum_{j=0}^{n-1}w_j\sum_{k=1}^{p}\frac{x_{j,k}}{2^k}
$$ 
is the $p$-binary approximation of~$\xx^\top\ww$.
\end{proposition}

\begin{proof}
We prove the proposition for the case $n=p=m=2$ for simplicity and it is clear that the general proof is analogous. Therefore we consider $
\Ug_{\ww,2} :=
\left\{\prod_{j=0}^{1}\prod_{k=1}^{2} \CR_z^{2,j,k}\left(\frac{w_j}{2^{2+k}}\right)\right\}^2 \prod_{j=0}^{1}\prod_{k=1}^{2} \CR_z^{1,j,k}\left(\frac{w_j}{2^{2+k}}\right).$
First, we have 
$$
 \prod_{j=0}^{1}\prod_{k=1}^{2} \CR_z^{1,j,k}\left(\frac{w_j}{2^{2+k}}\right)
 \otimes\left(\frac{1}{\sqrt{2^2}}\sum_{j=0}^{2^2-1}\ket{j}\right)
 \otimes\ket{\xx}
 =\frac{1}{\sqrt{2^2}}\left(\ket{0}+\ket{1}\right)\left(\ket{0}+\exp\left\{2\I\pi  \frac{\xxw^{\top}\ww}{2^2}\right\}\ket{1}\right)\otimes\ket{\xx},
$$
result to which we apply $\left\{\prod_{j=0}^{1}\prod_{k=1}^{2} \CR_z^{2,j,k}\left(\frac{w_j}{2^{2+k}}\right)\right\}^2$ which yields
$$
\frac{1}{\sqrt{2^2}}\left(\ket{0}+\exp\left\{2\I\pi 2 \frac{\xxw^{\top}\ww}{2^2}\right\}\ket{1}\right)
\otimes
\left(\ket{0}+\exp\left\{2\I\pi  \frac{\xxw^{\top}\ww}{2^2}\right\}\ket{1}\right)\otimes\ket{\xx},
$$
achieving the proof of~\eqref{comparison}.
\end{proof}
From the definition of the Quantum Fourier transform in~\eqref{QFT}, if $\xxw^{\top}\ww=k\in\{0,\ldots,2^m-1\}$, the resulting state is
$$
\Ug_{\ww,m}\left(\left(\Hg^{\otimes m}\ket{00}\right)\otimes\ket{\xx}\right)
= \left(\Ffq\ket{k}\right)\otimes\ket{\xx}
=\left(\Ffq\ket{\xxw^{\top} \ww}\right)\otimes\ket{\xx}.
$$
Thus only applying the Quantum Inverse Fourier Transform would be enough to retrieve $\ket{\xxw^{\top}\ww}$.
The pseudo-code is detailed in Algorithm~\ref{alg:QInnerProd} and the quantum circuit in the case $n=p=m=2$ is depicted in Figure~\ref{fig:QInnerProduct} (and detailed in Example~\ref{ex:Innerprod2}).

\begin{algorithm}
\caption{Quantum Inner Product (QIP)
$(\ww, \xx, \Ug_{\ww,m}, m, p, \eps)$}
\begin{algorithmic}
\State \textbf{Input:} 
    $\ww\in[-1,1]^n, \xx\in [0,1]^n$, $\eps>0$ the probability to mismeasure $\xx^\top\ww$, $\Ug_{\ww,m}$ unitary matrix, $m$ the number of ancilla qubits, $p$ the precision used for the binary fraction of each component of $\xx$, together with the constraint $\xx^\top\ww \in \{0,\ldots, 2^m-1\}$.
    \newline
    \State \textbf{Procedure:} \newline
    \State \textbf{1}.  $\ket{0}^{\otimes m}\ket{\xx}$
    \Comment{Initial state with $\ket{0}^{\otimes m}$ as ancilla register and $\ket{\xx}$ as data made of $n\times p$ qubits}
    \newline
    \State \textbf{2}. $\ket{0}^{\otimes m}\ket{\xx}\mapsto \Hg^{\otimes m}\ket{0}^{\otimes m}\ket{\xx} = \displaystyle \frac{1}{\sqrt{2^m}}\sum_{j=0}^{2^m-1}\ket{j}\ket{\xx}$
    \Comment{Apply Hadamard gates to the~$m$ ancillas register}
    \newline
    \State \textbf{3}. $\longrightarrow \displaystyle \frac{1}{\sqrt{2^m}}\sum_{j=0}^{2^m-1}
    \exp\left\{2\I\pi j \frac{\xxw^{\top} \ww}{2^m}\right\}
    \ket{j}\ket{\xx}$
    \Comment{Apply $\Ug_{\ww,m}$}
    \newline
    \State \textbf{4}. $\longrightarrow \ket{\xxw^{\top} \ww}\ket{\xx}$ 
    \Comment{Apply the inverse $\QFT$ where $\ket{\xxw^{\top} \ww}$ is a $p$-qubit approximation of $\xxw^{\top}\ww$}
    \newline
    \State \textbf{5}. $\longrightarrow \ket{\xxw^{\top} \ww}$ \Comment{Measure $\xxw^{\top} \ww$ with a probability at least $1-\eps$}
    \newline
    \Return \textbf{Output:} $\xxw^{\top} \ww$
\end{algorithmic}
\label{alg:QInnerProd}
\end{algorithm}

\begin{example}\label{ex:Innerprod2}
To understand the computations performed by the quantum gates,
consider the case where $n=p=2$. 
Therefore we only need $2\times2$ qubits to represent each element of the dataset which constitute the main register. 
Introduce an ancilla register  composed of $m=2$ qubits each initialised at~$\ket{0}$,
and suppose that the input state on the main register is $\ket{\xx}$. 
The goal here is then to encode as a phase the result of the inner product $\xx^{\top} \ww$ where $\ww=(w_0,w_1)^{\top}$. 
So in this exemple the entire wave function combining both the main register's qubits and the ancilla register's qubits is encoded in 6 qubits.
By denoting $\CR_z^{1,j,k}(\alpha)$ the $\CR_z(\alpha)$ matrix applied to the first qubit of the ancilla register and the qubit $\ket{x^{i}_{j,k}}$,
and $\CR_z^{2,j,k}(\alpha)$ the $\CR_z(\alpha)$ matrix applied to the second qubit of the ancilla register and the qubit $\ket{x_{j,k}}$.
Using the gates in~\eqref{eq:GatesUwm}, namely
$$
\Ug_{\ww,1} =
\prod_{j=0}^{1}\prod_{k=1}^{2} \CR_z^{1,j,k}\left(\frac{w_j}{2^{1+k}}\right) 
\quad\text{and}\quad
\Ug_{\ww,2} =
\left\{\prod_{j=0}^{1}\prod_{k=1}^{2} \CR_z^{2,j,k}\left(\frac{w_j}{2^{2+k}}\right)\right\}^2 \prod_{j=0}^{1}\prod_{k=1}^{2} \CR_z^{1,j,k}\left(\frac{w_j}{2^{2+k}}\right).$$

\begin{figure}
    \centering
    \begin{equation*}
    \Qcircuit @C=1.0em @R=0.2em @!R {
	 	\lstick{ \ket{{q}_{1}} :  } & \gate{\mathrm{H}} & \multigate{5}{\mathrm{U_{w,2}}} & \qw & \qswap & \gate{\mathrm{H}} & \ctrl{1} &  \qw & \qw & \qw & \qw & \meter & \qw & \qw & \qw & \qw\\
	 	\lstick{ \ket{{q}_{2}} :  } & \gate{\mathrm{H}} & \ghost{\mathrm{U_{w,2}}} & \qw & \qswap \qwx[-1]  & \qw & \multigate{1}{\mathrm{\gamma}} \qw & \qw & \qw & \qw & \gate{\mathrm{H}} & \qw & \qw & \meter & \qw & \qw \\
	 	\lstick{ \ket{{x}_{0,1}} :  } & \qw & \ghost{\mathrm{U_{w,2}}}  & \qw & \qw & \qw & \qw & \qw & \qw & \qw & \qw & \qw & \qw & \qw & \qw & \qw\\
	 	\lstick{ \ket{{x}_{0,2}} :  } & \qw & \ghost{\mathrm{U_{w,2}}} & \qw & \qw & \qw & \qw & \qw & \qw & \qw & \qw & \qw & \qw & \qw & \qw & \qw\\
	 	\lstick{ \ket{{x}_{1,1}} :  } & \qw & \ghost{\mathrm{U_{w,2}}} & \qw & \qw & \qw & \qw & \qw & \qw & \qw & \qw & \qw & \qw & \qw & \qw & \qw\\
	 	\lstick{ \ket{{x}_{1,2}} :  } & \qw & \ghost{\mathrm{U_{w,2}}} & \qw & \qw & \qw & \qw & \qw & \qw & \qw & \qw & \qw & \qw & \qw & \qw & \qw\\
	 	\lstick{c:} & \lstick{/_{_{2}}} \cw & \cw & \cw & \cw & \cw & \cw & \cw & \cw & \cw & \cw & \dstick{_{_{0}}} \cw \cwx[-6] & \cw & \dstick{_{_{1}}} \cw \cwx[-5] & \cw & \cw\\
	 }
\end{equation*}
    \caption{QIP circuit for $m=2$ ancilla qubits. The~$c$ line represents the classical register from which we retrieve the outcomes of the measurements. The controlled gate $\gamma$  performs as $ C(\gamma): \ket{q_1}\ket{q_2}\mapsto\ind_{\ket{q_1}=\ket{1}}(\ket{q_1})\ket{1}\otimes \E^{-\I\frac{\pi}{4}}\ket{q_2}+\ind_{\ket{q_1}=\ket{0}}(\ket{q_1})\ket{0}\otimes\ket{q_2}$}
\label{fig:QInnerProduct}
\end{figure}
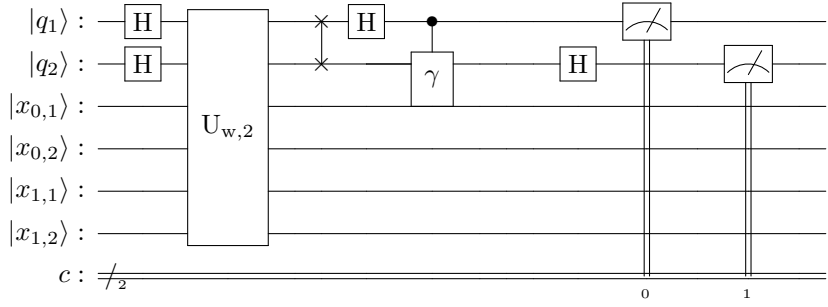
\end{example}

\begin{remark}
There is an interesting and potentially very useful difference here between the quantum and the classical versions of a feedforward neural network;
in the former, the input~$\xx$ is not lost after running the circuit, while this information is lost in the classical setting. 
This in particular implies that it can be used again for free in the quantum setting.
\end{remark}

\subsubsection{The case $\xx^{\top}\ww \notin \{0,\ldots, 2^m-1\}$.}

What happens if $\xxw^{\top} \ww$ is not a integer and $\xxw^{\top} \ww\geq 0$? 
Again, the short answer is that we are able to obtaina good approximation of~$\xxw^{\top} \ww$, 
which is already an approximation of the true value of the inner product~$\xx^{\top}\ww$. 
Indeed, with the gates constructed above, 
QIP performs exactly like QPE.
Just a quick comparison between what is obtained at stage 3 of the QPE Algorithm (Algorithm~\ref{algo:QPE}) and the output obtained at the third stage of the QIP~\eqref{comparison} would be enough to state that the QIP is just an application of the QPE procedure. 
Thus $\left\{\prod_{j=0}^{n-1}\prod_{k=1}^{p}\CR_z^{1,j,k}\left(\frac{w_j}{2^{m+k}}\right)\right\}$ is a unitary matrix such that $\ket{1}\otimes\ket{{\xx}}$ is an eigenvector of eigenvalue 
$\exp\left\{2\I\pi \frac{\xxw^{\top} \ww}{2^m}\right\}$.

Let $\phi:=\frac{1}{2^m}\xxw^{\top} \ww$;
the QPE procedure (Appendix~\ref{sec:QPE}) can only estimate $\phi \in [0,1)$. 
Firstly $\phi \leq 0$ can happen and secondly $\lvert \phi \rvert \geq 1$
can also happen. Therefore such circumstances have to be addressed. One first step would be to have
$\ww \in [-1,1]^n$,
so that $\lvert \xxw^{\top} \ww \rvert \leq n$. 
Then one should have $m$ (the number of ancillas) large enough so that
\begin{equation}\label{inequality param}
    \left| \frac{\xxw^{\top} \ww }{2^m}\right| \leq 1,
\end{equation}
which produces $m\geq \log_2(n)$. Having these constrains respected, 
one obtains
$|\phi|\leq 1$,
which is not enough since we should  have  $\phi \in [0,1)$ instead.
The main idea behind solving that is based on computing $\frac{\xxw^{\top} \ww }{2}$ instead of $\xxw^{\top} \ww$ which means dividing by 2 all the parameters of the $\CR_z^{m,j,k}$ gates.
Indeed with~\eqref{inequality param}, we have 
$-2^m \leq \xxw^{\top} \ww \leq 2^m$,
and thus
$-2^{m-1} \leq \half \xxw^{\top} \ww \leq 2^{m-1}$. 
\begin{itemize}
    \item In the case where $\xxw^{\top} \ww\geq 0$ we have
    $\frac{\xxw^{\top} \ww}{2} \in [0,2^{m-1}]$ and then by defining 
    $\widetilde{\phi}^{+}:=\frac{1}{2^m}\frac{\xxw^{\top} \ww}{2}$ 
we then obtain $\widetilde{\phi}^{+} \in [0,\frac{1}{2}]$, therefore the QPE can produce an approximation of $\widetilde{\phi}^{+}$ as put forward in Algorithm~\ref{algo:QPE}  which then can be multiplied by $2^{m+1}$ to retrieve $\xxw^{\top}\ww$ .
    \item In the case where $\xxw^{\top} \ww \leq 0$, 
    then $\frac{\xxw^{\top} \ww}{2} \in [-2^{m-1},0]$. 
    As above, $\ket{1}\otimes\ket{\xx}$ is an eigenvector of $\left\{\prod_{j=0}^{n-1}\prod_{k=1}^{p}\CR_z^{1,j,k}\left(\frac{\frac{w_j}{2}}{2^{m+k}}\right)\right\}$ with corresponding eigenvalue $\exp\left\{2\I\pi \frac{\frac{\xxw^{\top} \ww}{2}}{2^m}\right\}= \exp\left\{2\I\pi\left[1+ \frac{\frac{\xxw^{\top} \ww}{2}}{2^m}\right]\right\}$. 
   Defining 
   $\widetilde{\phi}^{-}
   := \frac{1}{2^m}\left(2^m+ \frac{\xxw^{\top} \ww}{2}\right)
   = 1+\frac{1}{2^m}\frac{\xxw^{\top} \ww}{2}$ we then obtain $\widetilde{\phi}^{-} \in [\half,1]$ which a QPE procedure can estimate 
   and from which we can retrieve $\xxw^{\top} \ww$
\end{itemize}
For values of~$\phi$ measured in $[0,\half) \cup (\half,1)$ we are sure about the associated value of the inner product. This means that for a fixed~$\xx$, the map
$$
f: \Big[0,\half\Big) \cup \left(\half,1\right)\ni \phi \mapsto \xxw^{\top} \ww \in [-n,n]
$$
is injective.
A measurement output equal to~$half$ could mean 
either that
$\xxw^{\top} \ww=2^m$ or $\xxw^{\top} \ww=-2^m$, 
which could be prevented for $\ww\in [-1,1]^n$ and~$m$ large enough such that $n < 2^m$. 
Under these circumstances, $f$ can be extended to an injective function on $[0,1)$, 
with~$1$ being excluded since the QPE can only estimate values in $[0,1)$.

\subsection{Quantum activation function}\label{sec:QActivation}
We consider an activation function~$\sigma:\RR\to\RR$.
A classical example is the sigmoid
$\sigma(x):=\left(1+\E^{-x}\right)^{-1}$.
The goal here is to build a circuit performing the transformation $\ket{x}\mapsto\ket{\sigma(x)}$ where $\ket{x}$ and $\ket{\sigma(x)}$ are the quantum encoded versions of their classical counterparts as in Section~\ref{sec:QEncoding}.
Again, we shall appeal to the Quantum Phase Estimation algorithm.
For a $q$-qubit state $\ket{x}=\ket{x_1\ldots x_q} \in \CC^{2^q}$, 
we wish to build a matrix~$\Ug \in \Mm_{2^q}(\CC)$ such that
$$
\Ug\ket{x}=\E^{2\I\pi \sigma(x)}\ket{x}.$$
Considering 
$$
\Ug := \text{Diag}\left(\E^{2\I\pi \sigma(0)},\E^{2\I\pi \sigma(1)},\E^{2\I\pi \sigma(2)},\ldots,\E^{2\I\pi \sigma(2^q-1)}\right),
$$
then, for~$m$ ancilla qubits, the Quantum Phase estimation yields
$$
\QPE: \ket{0}^{\otimes m}\otimes\ket{x}
\mapsto\ket{\widetilde{\sigma(x)}}\otimes\ket{x},
$$
where again $\widetilde{\sigma(x)}$ is the $m$-bit binary fraction approximation for~$\sigma(x)$ as detailed in Algorithm~\ref{algo:QPE}.
In Figure~\ref{fig:qNeuron}, we can see that the information flows from $\ket{\xx}=\ket{x_{0,1}x_{1,1}x_{2,1}x_{3,1}}$ to  the register attached to $\ket{q_2}$ to obtain the inner product and from the register~$\ket{q_2}$ to~$\ket{q_1}$ for the activation of the inner product. 
This explains why only measuring the register~$\ket{q_1}$ is enough to retrieve $\sigma(\xxw^{\top} \ww)$.

\begin{figure}[!h]
    \centering
\begin{equation*}
    \Qcircuit @C=1.5em @R=0.7em @!R {
	 	\lstick{ \ket{{q}_{1}}:  } & \gate{\mathrm{H}} & \qw & \qw & \multigate{1}{\mathrm{\sigma}} & \gate{\mathrm{H}} & \meter & \qw & \qw\\
	 	\lstick{ \ket{{q}_{2}}:  } & \gate{\mathrm{H}} & \multigate{4}{\mathrm{\Ug_{\ww,1}}} & \gate{\mathrm{H}} & \ghost{\mathrm{\sigma}} & \qw & \qw & \qw & \qw\\
	 	\lstick{ \ket{{x}_{0,1}}:  } & \qw & \ghost{\mathrm{\Ug_{\ww,1}}} & \qw & \qw & \qw & \qw & \qw & \qw\\
	 	\lstick{ \ket{{x}_{1,1}}:  } & \qw & \ghost{\mathrm{\Ug_{\ww,1}}} & \qw & \qw & \qw & \qw & \qw & \qw\\
	 	\lstick{ \ket{{x}_{2,1}}:  } & \qw & \ghost{\mathrm{\Ug_{\ww,1}}} & \qw & \qw & \qw & \qw & \qw & \qw\\
	 	\lstick{ \ket{{x}_{3,1}}:  } & \qw & \ghost{\mathrm{\Ug_{\ww,1}}} & \qw & \qw & \qw & \qw & \qw & \qw\\
	 	\lstick{c:} & \lstick{/_{_{1}}} \cw & \cw & \cw & \cw & \cw & \dstick{_{_{0}}} \cw \cwx[-6] & \cw & \cw\\}
\end{equation*}
    \caption{Quantum single neuron for $\ket{\xx} \in \mathbb{C}^{2^4},$ one ancilla qubit $\ket{q_2}$ for the QIP implemented via the controlled gate $\Ug_{\ww,1}$ for $w \in [-1,1]^4$, and one ancilla qubit $\ket{q_1}$ for the activation function~$\sigma$.}
    
    \label{fig:qNeuron}
\end{figure}
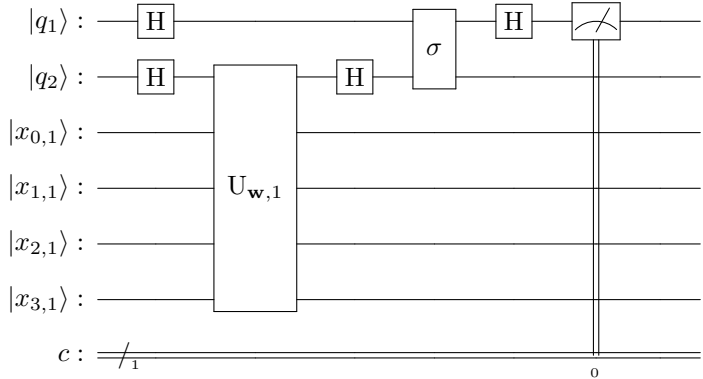

\section{Quantum GAN architecture}\label{sec:QuGAN}

A Generative Adversarial Network (GAN) is a network composed of two neural networks.
In a classical setting, two agents, the generator and the discriminator, compete against each other in a zero-sum game~\cite{GameTheory}, 
playing in turns to improve their own strategy;
the generator tries to fool the discriminator 
while the latter aims at correctly distinguishing real data (from a training database) from generated ones. 
As put forward in~\cite{goodfellow2014generative}, 
the generative model can be thought of as an analogue to a team of counterfeiters,
trying to produce fake currency and use it without detection, while the discriminator plays the role of the police, trying to detect the counterfeit currency. Competition in this game drives
both teams to improve their methods until the counterfeits are indistinguishable from the genuine
articles. 
Under reasonable assumptions 
(the strategy spaces of the agents are compact and convex)
the game has a unique (Nash) equilibrium point, 
where the generator is able to reproduce exactly the target data distribution.
Therefore, in a classical setting, the generator~$\Gf$, parameterised by a vector of parameters~$\thbg$, produces a random variable~$X_{\thbg}$, which we can write as the map
$$
\Gf: \thbg \xrightarrow[]{} X_{\thbg}.
$$
The goal of the discriminator~$\Df$, parameterised by~$\thbd$, is to distinguish samples~$\xx_{\thbg}$ of~$X_{\thbg}$ from $\xx_{\real} \in \Dd$, 
where $\xx_{\real}$ has been sampled from the underlying distribution $\PP_{\Dd}$ of the database~$\Dd$. 
The map $\Df$ thus reads
$$
\Df: \xx_{\thbg},\thbd \mapsto \PP_{\thbd}\left(\xx_{\thbg} \text{ sampled from }  \PP_{\Dd}\right).
$$
We aim here at mimicking this classical GAN architecture into quantum version.
Not surprisingly, we first build a quantum discriminator, followed by a quantum version of the generator, 
and we finally develop the quantum equivalent of the zero-sum game, defining an objective loss function acting on quantum states.

\subsection{Quantum discriminator}

In the case of a fully connected quantum GAN - which we study here - where both the discriminator and generator are quantum circuits, one of the main differences between a classical GAN and a QuGAN lays in the input of the discriminator.
Indeed, as said above, in a classical discriminator the input is a sample $\xx_{\thbg}$ generated by the generator~$\Gf$, whereas in a quantum discriminator the input is a wave function 
\begin{equation}\label{eq:WaveFunction}
\ket{v_{\thbg}}=\sum_{j=0}^{2^n-1}v_{j,\thbg}\ket{j}
\end{equation}
generated by a quantum generator. 
In such a setting, the goal is to create a wave function of the form~\eqref{eq:WaveFunction}
which is a physical way of encoding a given discrete distribution, namely
\begin{equation} \label{eq:proba 2}
\PP\left(\ket{v_{\thbg}}=\ket{j}\right) = |v_{j,\thbg}|^2=p_j,
\qquad\text{fo each } j=0,\ldots, 2^n-1,
\end{equation}
where $(p_{j})_{j=0,\ldots, 2^n-1} \in [0,1]^{2^n}$ with $\sum_{j=0}^{2^n-1}p_{j}=1$.
We choose here a simple architecture for the discriminator, as a quantum version 
of a perceptron with a sigmoid activation function (Figure~\ref{fig:ClassicalPerceptron}).

\begin{figure}[h!]\label{fig:Perceptron}
    \centering
    \tikzset{basic/.style={draw,fill=none,
                       text badly centered,minimum width=3em}}
\tikzset{input/.style={basic,circle,minimum width=2.5em}}
\tikzset{weights/.style={basic,rectangle,minimum width=2em}}
\tikzset{functions/.style={basic,circle, minimum width=3.5em}}
\newcommand{\addaxes}{\draw (0em,1em) -- (0em,-1em)
                            (-1em,0em) -- (1em,0em);}
\newcommand{\sigm}{\draw[line width=1.5pt] (-1em,0) -- (0,0)
                                (0,0) -- (0.75em,0.75em);}
\newcommand{\stepfunc}{\draw[line width=1.5pt] (0.65em,0.65em) -- (0,0.65em) 
                                    -- (0,-0.65em) -- (-0.65em,-0.65em);}

    \begin{tikzpicture}[scale=0.8]
    \foreach \h [count=\hi ] in {$x_2$,$x_1$,$x_0$}{%
          \node[input] (f\hi) at (0,\hi*1.25cm-1.5 cm) {\h};
        }
    \node[below=0.62cm] (idots) at (f1) {\vdots};
    \node[input, below=0.62cm] (last_input) at (idots) {$x_{n-1}$};
    \node[functions] (sum) at (5,0) {$\xx^\top\ww$};
    \foreach \h [count=\hi ] in {$w_2$,$w_1$,$w_0$}{%
          \path (f\hi) -- node[weights] (w\hi) {\h} (sum);
          \draw[->] (f\hi) -- (w\hi);
          \draw[->] (w\hi) -- (sum);
        }
    \node[below=0.05cm] (wdots) at (w1) {\vdots};
    \node[weights, below=0.45cm] (last_weight) at (wdots) {$w_{n-1}$};
    \path[draw,->] (last_input) -- (last_weight);
    \path[draw,->] (last_weight) -- (sum);
    \node[functions] (activation) at (8,0) {$\huge\sigma$};
    \draw[->] (sum) -- (activation);
    \draw[->] (activation) -- ++(2,0);
    \node at (0,4) {inputs};
    \node at (3,4) {weights};
    \node at (8,4) {activation function};
    \end{tikzpicture}

    \caption{Classical perceptron mapping $\xx\in\RR^{n}$ to $\sigma\left(\xx^\top\ww\right) \in \RR$.}
    \label{fig:ClassicalPerceptron}
\end{figure}
This approach of building the circuit is new since in the papers that use  quantum discriminators, the circuits that are used are what is called ansatz circuits~\cite{HowToEnhance}, 
in other words generic circuits built with layers of rotation gates and controlled rotation gates (see~\eqref{eq:RYGate} and~\eqref{eq:CRYGate} below for the definition of these gates). 
Such ansatz circuits are therefore parameterised circuits as put forward in~\cite{chakrabarti2019quantum}, where generally an interpretation on the circuit's architecture performing as a classifying neural network cannot be made. As pointed out in \cite{HowToEnhance}, the architectures of both the generator and the discriminator are the same, which on the one hand solves the issue of having to monitor whether there is a imbalance in terms of expressivity between the generator and the discriminator, however on the other hand it prevents us from being able to give a straightforward interpretation for the given architectures. 

The main task here is then to translate these classical computations to a quantum input for the discriminator. 
This challenge has been taken up in both~\ref{sec:QInnerProd} and~\ref{sec:QActivation}  where we have built from scratch a quantum perceptron which performs exactly like a classical perceptron. 
There is however one main difference in terms of interpretation:
let the wave function~\eqref{eq:WaveFunction} be the input for the discriminator with $N=2^n$ and, 
for $j = \overline{j_1\cdots j_n}$ (defined in~\eqref{eq:binRepr}), 
define $\phi_{j}:=(j_1,\ldots,j_n)$.
Denote $\Dr(\ww) \in \Mm_{2^{n+m_{1}+m_{2}}}(\CC)$ the transformation performed by the entire quantum circuit depicted in Figure~\ref{fig:QuantumPerceptron}, where $\Dr(\ww)$ is unitary and 
$\ww\in \RR^n$, namely for $m_{1}+m_{2}$ ancilla qubits, 
$$
\Dr(\ww)\ket{0}^{\otimes{m_{1}+m_{2}}}\ket{j}
 = \ket{\sigma\left(\phi_{j}^{\top} \ww\right)}\ket{\phi_{j}^{\top} \ww)}\ket{j},
$$
where $\ket{\sigma\left(\phi_{j}^{\top} \ww\right)} \in \CC^{2^{m_{1}}}$ 
and $\ket{\phi_{j}^{\top} \ww} \in \CC^{2^{m_{2}}}$
and where we only measure $\ket{\sigma\left(\phi_{j}^{\top} \ww\right)}$.
Thus, for the input~\eqref{eq:WaveFunction}, 
the discriminator outputs the wave function (with $m_{1}+m_{2}$ ancilla qubits)
\begin{equation}\label{Output Discrim}
\Dr(\ww)\ket{0}^{\otimes{m_{1}+m_{2}}}\ket{v_{\thbg}}
 = \sum_{j=0}^{2^n-1}v_{j,\thbg}\ket{\sigma\left(\phi_{j}^{\top} \ww\right)}\ket{\phi_{j}^{\top} \ww)}\ket{j}.
\end{equation}
Therefore, in a QuGAN setting the goal for the discriminator is to distinguish the target wave function $\ket{\psi_{\tar}}$ from the generated one $\ket{v_{\thbg}}$.

\begin{example}\label{ex:Discrimin}
As an example, consider $m_{2}=1$ ancilla qubit for the inner product, $m_{1}=1$ ancilla qubit for the activation, 
$\ket{\psi_{\tar}}=\psi_0\ket{0}+\psi_1\ket{1}$ and $\ket{v_{\thbg}}=v_{0,\thbg}\ket{0}+v_{1,\thbg}\ket{1}$.
As we only measure the outcome produced by the activation function, 
the only possible outcomes are~$\ket{0}$ and~$\ket{1}$. 
Therefore, measuring the output of the discriminator only consists of a projection on either $\ket{0}$ or $\ket{1}$. 
Define these projectors
$$
\Pi_0 := \ket{0}\bra{0}\otimes \Id^{\otimes m_{2}+n} \in \mathcal{M}_{2^{m_{1}+n+m_{2}}}(\mathbb{C})
\qquad\text{and}\qquad
\Pi_1 := \ket{1}\bra{1}\otimes \Id^{\otimes m_{2}+n} \in \mathcal{M}_{2^{m_{1}+n+m_{2}}}(\mathbb{C}),
$$ 
where $m_{2}=1$ and $n=1$ since in our toy example the wave functions encoding the distributions are 
$1$-qubit distributions. 
Interpreting measuring~$\ket{0}$ as labelling the input distribution \textit{Fake} 
and measuring~$\ket{1}$ as labelling it \textit{Real}, 
the optimal discriminator with parameter~$\ww^*$ would perform as 
\begin{equation}\label{eq:ProbaProj}
\begin{array}{rclr}
\displaystyle \PP\left( \Dr(\ww^*)\ket{0}^{\otimes{m_{1}+m_{2}}}\ket{v_{\thbg}}
 =\ket{0}\otimes\sum_{j=0}^{2^n-1}v_{j,\thbg}\ket{\phi_{j}^{\top} \ww^*}\ket{j}\right)
 & = & \displaystyle \left\|\Pi_0\Dr(\ww^*)\ket{0}^{\otimes{m_{1}+m_{2}}}\ket{v_{\thbg}}\right\|^2
& =1,\\
\displaystyle \PP\left(\Dr(\ww^*)\ket{0}^{\otimes{m_{1}+m_{2}}}\ket{\psi_{\tar}}
=\ket{1}\otimes\sum_{j=0}^{2^n-1}\psi_{j}\ket{\phi_{j}^{\top} \ww^*}\ket{j}\right)
 & = & \displaystyle \left\|\Pi_1\Dr(\ww^*)\ket{0}^{\otimes{m_{1}+m_{2}}}\ket{\psi_{\tar}}\right\|^2
& =1,
\end{array}
\end{equation}
where still in our toy example we have $n=1$, $m_{1}=1$ and $m_{2}=1$. Here $n$ could be any positive integer.
We illustrate the circuit in Figure~\ref{fig:QuantumPerceptron}.

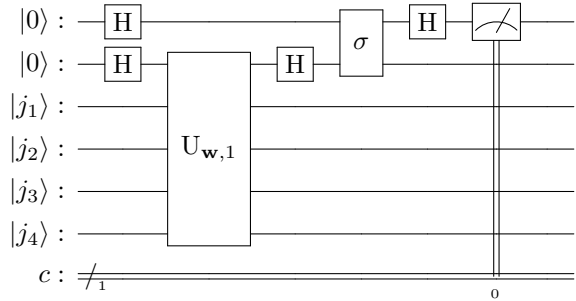
\begin{figure}[h!]\label{QDiscriminator}
    \centering
    \begin{equation*}
    \Qcircuit @C=1.0em @R=0.2em @!R {
	 	\lstick{ \ket{0}:  } & \gate{\mathrm{H}} & \qw & \qw & \multigate{1}{\mathrm{\sigma}} & \gate{\mathrm{H}} & \meter & \qw & \qw\\
	 	\lstick{ \ket{0}:  } & \gate{\mathrm{H}} & \multigate{4}{\mathrm{\Ug_{\ww,1}}} & \gate{\mathrm{H}} & \ghost{\mathrm{\sigma}} & \qw & \qw & \qw & \qw\\
	 	\lstick{ \ket{{j}_{1}}:  } & \qw & \ghost{\mathrm{\Ug_{\ww,1}}} & \qw & \qw & \qw & \qw & \qw & \qw\\
	 	\lstick{ \ket{{j}_{2}}:  } & \qw & \ghost{\mathrm{\Ug_{\ww,1}}} & \qw & \qw & \qw & \qw & \qw & \qw\\
	 	\lstick{ \ket{{j}_{3}}:  } & \qw & \ghost{\mathrm{\Ug_{\ww,1}}} & \qw & \qw & \qw & \qw & \qw & \qw\\
	 	\lstick{ \ket{{j}_{4}}:  } & \qw & \ghost{\mathrm{\Ug_{\ww,1}}} & \qw & \qw & \qw & \qw & \qw & \qw\\
	 	\lstick{c:} & \lstick{/_{_{1}}} \cw & \cw & \cw & \cw & \cw & \dstick{_{_{0}}} \cw \cwx[-6] & \cw & \cw\\
	 }
\end{equation*}
    \caption{Quantum perceptron with $\ww \in \RR^4$ and one ancilla qubit 
    for the inner product ($m_{2}=1$) and one ancilla qubit for the activation ($m_{1}=1$). 
    Here we only measure the result produced by the activation function.}
\label{fig:QuantumPerceptron}
\end{figure}
\end{example}

\subsubsection{Bloch sphere representation}
The Bloch sphere~\cite{nielsen00} is an important in Quantum Computing, 
providing a geometrical representation of pure states. 
In our case, it yields a geometric visualisation of the way an optimal quantum discriminator works
as it separates the two complementary regions 
\begin{equation}\label{eq:Regions}
\begin{array}{rl}
\Rr_F & := \displaystyle \left\{\sum_{i=0}^{2^{m-1}
-1}\alpha_i\ket{i} \text{ such that } \sum_{i=0}^{2^{m-1}
-1}|\alpha_i|^2=1\right\},\\
\Rr_T & := \displaystyle \left\{\sum_{i=2^{m-1}}^{2^{m}
-1}\alpha_i\ket{i} \text{ such that } \sum_{i=2^{m-1}}^{2^{m}
-1}|\alpha_i|^2=1\right\},
\end{array}
\end{equation}
where $m:=m_{1}+m_{2}+n$ is the total number of qubits for the inputs of the discriminator. 
The optimal discriminator $\Dr(\ww^*)$ would perform as 
$$
\Dr(\ww^*)\ket{\fake} \in \Rr_F
\quad\text{and}\quad
\Dr(\ww^*)\ket{\real} \in \Rr_T, 
\quad\text{almost surely},
$$
where $\ket{\fake}:=\ket{0}\ket{0}\ket{v_{\thbg}}$ and
$\ket{\real}:=\ket{0}\ket{0}\ket{\psi_{\tar}}$. 
Now, the challenge lays in finding such an optimal discriminator, however one should note that the nature of the state $\ket{\fake}$ plays a major role in finding such a discriminator.
Therefore, in the following part we focus on the generator responsible for generating $\ket{\fake}$.

\begin{example}
Consider Example~\ref{ex:Discrimin} with 
$(\psi_0, \psi_1) = (\frac{1}{\sqrt{2}}, \frac{1}{\sqrt{2}})$ 
and $(v_{0,\thbg}, v_{1,\thbg}) = (\frac{\sqrt{3}}{2}, \frac{1}{2})$.
The states~$\ket{\psi_{\tar}}$ and~$\ket{v_{\thbg}}$ are shown in Figure~\ref{fig: Bloch Sphere}.
The wave function produced by the discriminator is composed of three qubits 
($m_{1}=1$, $m_{2}=1$ and $n=1$ qubit for the input wave function~\eqref{Output Discrim}), 
therefore one optimal transformation for the discriminator having $\ket{\psi_{\tar}}$ as an input is one such that the first qubit never collapses onto the state~$\ket{0}$.
\begin{figure}[h!]
    \centering
    \includegraphics[width=5cm]{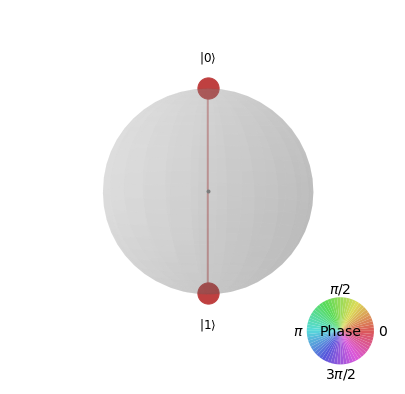}
    \qquad
    \includegraphics[width=5cm]{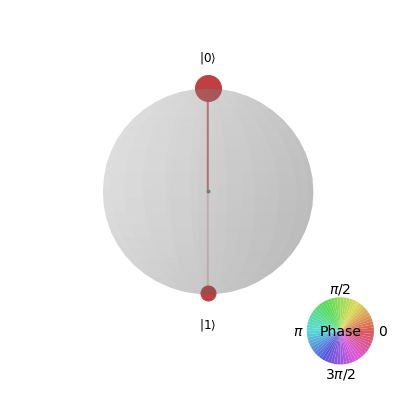}
    \caption{Bloch spheres representations for $\ket{\psi_{\tar}}$ (left) and $\ket{v_{\thbg}}$ (right) where there is no phase shift between~$\ket{0}$ and~$\ket{1}$
    and where the sizes of the lobes are proportional to the probability of measuring the associated states.}
    \label{fig: Bloch Sphere}
\end{figure}
\end{example}

\begin{figure}[h!]
    \centering
    \includegraphics[width=4.5cm]{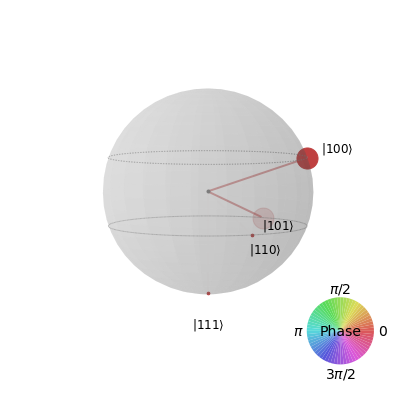}
    \centering
    \includegraphics[width=4.5cm]{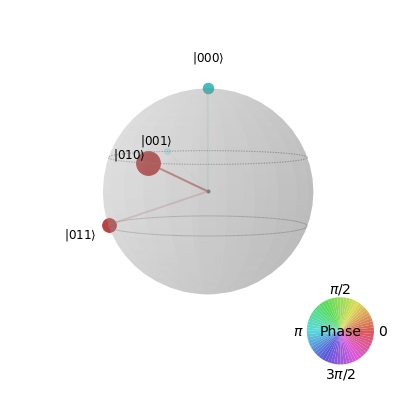}
    \caption{Left: $\Dr(w^*_1)\ket{0}\ket{0}\ket{\psi_{\tar}}$.
    Total system post one optimal discriminator transformation.
    As we can see here the first qubit never collapses onto $\ket{0}$ and therefore such a discriminator is optimal at labelling $\ket{\psi_{\tar}}$ as \textit{Real}.
    Right: $\Dr(w^*_2)\ket{0}\ket{0}\ket{v_{\thbg}}$.
    Total system post one optimal discriminator transformation.
    As we can see here the first qubit never collapses onto $\ket{1}$ and therefore such a discriminator is optimal at labelling $\ket{v_{\thbg}}$ \textit{Fake}.}
\end{figure}

\subsection{Quantum generator}

The quantum generator is a quantum circuit producing a wave function that encodes a discrete distribution. 
Such a circuit takes as an input the ground state $\ket{0}\otimes^{n-m_{1}-m_{2}}$
and outputs a wave function $\ket{v_{\thbg}}$ parameterised by $\thbg$, the set of parameters for the discriminator. 
We recall here a few quantum gates that will be key to constructing a quantum generator.
Recall that a quantum gate can be viewed as a unitary matrix;
of particular interest will be gates acting on two (or more) qubits, as its allows quantum entanglement,
thus fully leveraging the power of quantum computing.
The $\textrm{NOT}$ gate~$\Xg$ acts on one qubit and is represented as
$$
\Xg = \begin{pmatrix}
0 & 1 \\
1 & 0 
\end{pmatrix},
$$
so that $\Xg\ket{0} = \ket{1}$ and $\Xg\ket{1} = \ket{0}$.
The $\RY$ is a one-qubit gate represented by the matrix
\begin{equation}\label{eq:RYGate}
\RY(\theta)
 := \begin{pmatrix}
\cos\left(\frac{\theta}{2}\right) & -\sin\left(\frac{\theta}{2}\right)\\
\sin\left(\frac{\theta}{2}\right) & \cos\left(\frac{\theta}{2}\right) 
\end{pmatrix},
\end{equation}
thus performing as
$$
\RY(\theta)\ket{0} = \cos\left(\frac{\theta}{2}\right)\ket{0} + \sin\left(\frac{\theta}{2}\right)\ket{1}
\qquad\text{and}\qquad
\RY(\theta)\ket{1} = \cos\left(\frac{\theta}{2}\right)\ket{1} - \sin\left(\frac{\theta}{2}\right)\ket{0}.
$$

The $\CRY$ Gate is the controlled version of the~$\RY$ gate, acting on two qubits, one control qubit and one transformed qubit, producing quantum entanglement. 
The~$\RY$ transformation applies on the second qubit only when provided the control qubit is in~$\ket{1}$, 
otherwise leaves the second qubit unaltered.
Its matrix representation is
\begin{equation}\label{eq:CRYGate}
\CRY(\theta)=\begin{pmatrix}
1 & 0 & 0 & 0  \\
0 & 1 & 0 & 0 \\
0 & 0 & \cos\left(\frac{\theta}{2}\right) & -\sin\left(\frac{\theta}{2}\right)\\
0 & 0 & \sin\left(\frac{\theta}{2}\right) & \cos\left(\frac{\theta}{2}\right).
\end{pmatrix}
\end{equation}

Given~$n$ qubits
let $\XX:=(X_1\ldots X_n)$ be a random vector taking values in $\Xx_n := \{0, 1\}^n$. 
Set $$
p_{\xx} := \mathbb{P}[\XX = \xx], \quad\text{for } \xx\in \Xx_n.
$$
When building the generator we are looking for a quantum circuit that implements the transformation
$$
\ket{0}^{\otimes n}\mapsto\sum_{\xx\in \{0,1\}^n}\sqrt{p_{\xx}}\E^{{\I\theta_{\xx}}}\ket{\xx}.
$$

We could follow a classical algorithm. 
For $1 \leq k \leq n$, let $\xx_{:k}:=(x_1,\ldots, x_k)$ and,
given $\xx\in\Xx_n$, 
\begin{equation}
q_{\xx_{:k}} := 
\left\{
\begin{array}{ll}
\mathbb{P}[X_1 = 0], & \text{if }k=1,\\
\mathbb{P}[X_k = 0|\XX_{:k-1} = \xx_{:k-1}], & \text{if }2\leq k\leq n.
\end{array}
\right.
\end{equation}
We then proceed by induction: 
start with a random draw of~$\XX_1$ as a Bernoulli sample with failure probability~$q_{\xx_1}$.
Assuming that~$\XX_{:k-1}$ has been sampled as~$\xx_{:k-1}$ for some $1 \leq k \leq n$, 
sample~$\XX_k$ from a Bernoulli distribution with failure probability $q_{\xx_{:k-1}}$. 
The quantum circuit will equivalently consist of~$n$ stages, 
where at each stage $1 \leq k \leq n$ we only work with the first~$k$ qubits, 
and at the end of each stage there is the correct distribution for the first~$k$ qubits 
in the sense that, upon measuring, their distribution coincides with that of~$\XX_{:k}$.

The first step is simple: a single $\Yg$-rotation of the first qubit with angle 
$\theta \in [0, \pi]$ satisfying $\cos(\frac{\theta}{2}) = \sqrt{q_{\xx_1}}$. 
In other words, with $\Ug_{1} := \RY(\theta)$, 
we map $\ket{0}$ to $\Ug_{1}\ket{0} = \sqrt{q_{\xx_1}}\ket{0} + \sqrt{1-q_{\xx_1}}\ket{1}.$ 
Clearly, when measuring the first qubit, we obtain the correct law.
Now, inductively, for $2 \leq k \leq n$, 
suppose the first~$k-1$ qubits fixed, namely in the state
$$
\sum_{\xx_{:k-1}\in\Xx_{k-1}}\sqrt{p_{\xx_{:k-1}}}\ket{\xx_{:k-1}}\ket{0}^{\otimes n-k+1},
$$
For each $\xx_{:k-1}\in \Xx_{k-1}$, let $\theta_{\xx_{:k-1}}\in[0;\pi]$ satisfy $\cos\left(\half \theta_{\xx_{:k-1}}\right)=\sqrt{q_{\xx_{:k-1}}}$ 
and consider the gate $\Cg_{\xx_{:k-1}}$ acting on the first~$k$ qubits 
which is a $\RY(\theta_x)$ on the last qubit~$k$, controlled on whether the first $k-1$ qubits are equal to~$\xx_{:k-1}$.
We then have
have
\begin{equation}
C_{\xx_{:k-1}}\ket{\yy}\ket{0} = 
\left\{\begin{array}{ll}
\displaystyle \sqrt{q_{\xx_{:k-1}}}\ket{\xx_{:k-1}}\ket{0} + \sqrt{1-q_{\xx_{:k-1}}}\ket{\xx_{:k-1}}\ket{1}, & \quad \text{if } \yy = \xx_{:k-1},\\
\displaystyle \ket{\yy}\ket{0}, \quad\text{for }\yy \ne \xx_{:k-1}.
\end{array}
\right.
\end{equation}
Therefore, defining $\Ug_k := \prod_{\xx_{:k-1}\in\Xx_{k-1}}\Cg_{\xx_{:k-1}}$, 
and noting that the order of multiplication does not affect the computations below, it follows that
\begin{align*}
\Ug_k\sum_{\xx_{:k-1}\in\Xx_{k-1}}\sqrt{p_{\xx_{:k-1}}}\ket{\xx_{:k-1}}\ket{0}^{\otimes n-k+1}
 & = \sum_{\xx_{:k-1}\in\Xx_{k-1}}\left\{\sqrt{p_{\xx_{:k-1}}q_{\xx_{:k-1}}}\ket{\xx_{:k-1}}
 +\sqrt{p_{\xx_{:k-1}}\left(1-q_{\xx_{:k-1}}\right)}\ket{1}\right\}\ket{0}\\
 & \sum_{\xx_{:k}\in\Xx_{k}}\sqrt{p_{\xx_{:k}}}\ket{\xx_{:k}}\ket{0}^{\otimes n-k},
\end{align*}
where the last equality follows from properties of conditional expectations
since 
$$
p_{\xx_{:k-1}} q_{\xx_{:k-1}} = p_{{\xx_{:k-1}}.0}
\qquad\text{and}\qquad
p_{\xx_{:k-1}}\left(1-q_{\xx_{:k-1}}\right)=p_{{\xx_{:k-1}}.1},
$$
for ${\xx_{:k-1}}\in \Xx_{k-1}$, ${\xx_{:k-1}}.0 \in \Xx_{k}$ and ${\xx_{:k-1}}.1 \in \Xx_{k}$
(see after~\eqref{eq:binRepr} for the binary representation of decimals).
This concludes the inductive step.
The generator has therefore been built accordingly to a `classical' algorithm,
however only up until~$\Xx_2$ 
(see Figure~\ref{fig:QGen} for the architecture for qubits~$q_3$ and~$q_2$) 
to avoid to have a network that is too deep and therefore untrainable in a differentiable manner because of the barren plateau phenomenon~\cite{BarrenPlateau}. 
Indeed, in order to build~$\Ug_k$ from simple controlled gates (with only one control qubit) 
the number of gates is of order~$\Oo(2^{k-1})$, making the generator deeper. 
Thus the number of gates we would have to use would be of order~$\Oo(2^{n})$, 
making the generator very expressive yet very hard to train.

\begin{example}
With $n=4$, the architecture for our generator is depicted in Figure~\ref{fig:QGen}
and the full QuGAN (generator and discriminator) algorithm in Figure~\ref{fig:FullQuGAN}.

\begin{figure}[h!]
    \centering
    \includegraphics[scale=0.5]{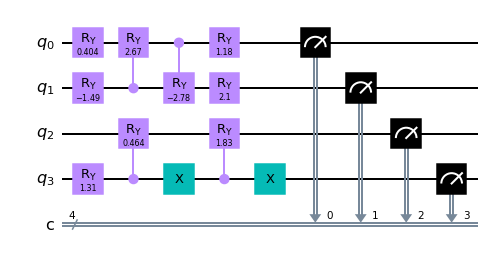}
    \caption{Entangled generator composed of $\RY$, $\CRY$ and $\Xg$ gates, with parameters values for $\{\theta_1,\ldots,\theta_{9}\}$ indicated alongside the gates.}
    \label{fig:QGen}
\end{figure}

\begin{figure}
    \centering
\includegraphics[scale=.5]{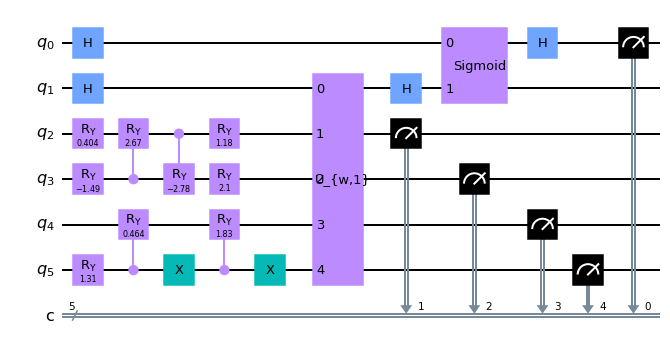}
\caption{The entire associated entangled QuGAN.}
\label{fig:FullQuGAN}
\end{figure}
\end{example}

\subsection{Quantum adversarial game}

In GANs the goal of the discriminator~(D) is to discriminate real~(R) data from the fake ones generated by the generator~(G), while the goal of the latter is to fool the discriminator by generating fake data. 
Here both real and generated data are modeled as quantum states, respectively described by their wave functions
$\ket{\psi_{\tar}}$ and $\ket{v_{\thbg}}$. 
Define the objective function
$$
\Ss(\thbg,\wwd)
:= \PP\Big(\Dr(\wwd)\ket{0}\ket{0}\ket{\psi_{\tar}}\in \Rr_{T}\Big)
 - \PP\Big(\Dr(\wwd)\ket{0}\ket{0}\ket{v_{\thbg}}\in \Rr_{T}\Big),
$$
where the region~$\Rr$ is defined in~\eqref{eq:Regions}.
Here
$\PP(\Dr(\wwd)\ket{0}\ket{0}\ket{\psi_{\tar}}\in \Rr_{T})$
is the probability of labelling the real data $\ket{0}\ket{0}\ket{\psi_{\tar}}$ as real via the discriminator and 
$\PP(\Dr(\wwd)\ket{0}\ket{0}\ket{v_{\thbg}}\in \Rr_{T})$ is the probability of having the generator fool the discriminator. As stated in~\eqref{eq:ProbaProj} for two ancilla qubits ($m_{1}+m_{2}=2$, i.e. one qubit for inner product and one qubit for activation) we have
$$
\PP\Big(\Dr(\wwd)\ket{0}\ket{0}\ket{\psi_{\tar}}\in \Rr_{T}\Big)
= \left\|\Pi_1\Dr(\wwd)\ket{0}\ket{0}\ket{\psi_{\tar}}\right\|^2.
$$
By defining the projection of the output of the discriminator onto $\Rr_T$,
$$
\ket{\psi_{\out,\tar,\wwd}} := \Pi_1\Dr(\wwd)\ket{0}\ket{0}\ket{\psi_{\tar}},
$$ 
we can also write 
$$
\PP\Big(\Dr(\wwd)\ket{0}\ket{0}\ket{\psi_{\tar}}\in \Rr_{T}\Big) = \Tr( \rho_{\out,\tar,\wwd}),
$$ 
where $\rho_{\out,\tar,\wwd}:=\ket{\psi_{\out,\tar,\wwd}}\bra{\psi_{\out,\tar,\wwd}}$
is the density operator associated to $\psi_{\out,\tar,\wwd}$.
The same goes for the probability of fooling the discriminator, namely
$$
\PP\Big(\Dr(\wwd)\ket{0}\ket{0}\ket{v_{\thbg}}\in \Rr_{T}\Big)
=
\left\|\Pi_1\Dr(\wwd)\ket{0}\ket{0}\ket{v_{\thbg}}\right\|^2
=\Tr(\rho_{\out,\thbg,\wwd}),
$$
where 
$\ket{\psi_{\out,\thbg,\wwd}}:=\Pi_1\Dr(\wwd)\ket{0}\ket{0}\ket{v_{\thbg}}$
and 
$\rho_{\out,\thbg,\wwd}:=\ket{\psi_{\out,\thbg,\wwd}}\bra{\psi_{\out,\thbg,\wwd}}$.
The min-max game played by the Generative Adversarial network is therefore defined as the optimisation problem
\begin{equation}\label{eq:MinMax}
    \min_{\thbg}\max_{\wwd} \Ss(\thbg,\wwd).
\end{equation}

Moreover, since~$\Ss$ is differentiable and given the architecture of our circuits, according to the shift rule formula~\cite{ShiftRule}, 
the partial derivatives of~$\Ss$ admit the closed-form representations
\begin{equation}\label{partialderiv1}
\begin{array}{rl}
\displaystyle \nabla_{\thbg}\Ss(\thbg,\wwd)
 & = \displaystyle \half\left\{
 \Ss\left(\thbg+\frac{\pi}{2},\wwd\right)
 - \Ss\left(\thbg-\frac{\pi}{2},\wwd\right)\right\},\\
\displaystyle \nabla_{\wwd}\Ss(\thbg,\wwd)
 &  = \displaystyle \half\left\{
 \Ss\left(\thbg,\wwd+\frac{\pi}{2}\right)
  - \Ss\left(\thbg,\wwd-\frac{\pi}{2}\right)\right\},
\end{array}
\end{equation}
so that training will be based on stochastic gradient ascent and descent.
The reason for a stochastic algorithm lies in the nature of~$\Ss(\thbg,\wwd)$,
seen as the difference between two probabilities to estimate. 
A natural estimator for~$l$ measurements/observations is
$$
\widehat{\Ss}(\thbg,\wwd)_l
:= \frac{1}{l}\sum_{k=1}^{l}
\ind_{\left\{\Dr(\wwd)\ket{0}\ket{0}\ket{\psi_{\tar}^k}\in \Rr_{T}\right\}}
 - \ind_{\left\{\Dr(\wwd)\ket{0}\ket{0}\ket{v^k_{\thbg}}\in \Rr_{T}\right\}},
 $$
where $\ket{v_{\thbg}^k}$ is the $k$-th wave function produced by the generator and $\ket{\psi_{\tar}^k}$ is the $k$-th copy for the target distribution.

Given the nature of the problem, two strategies arise:
for fixed parameters~$\thbg$, when training the discriminator, we first minimise the labelling error, ie.
$$
\max_{\wwd}\Ss(\thbg,\wwd),
$$
which can be achieved by stochastic gradient ascent.
Then, when training the generator the goal is to fool the discriminator, so that,  for fixed~$\wwd$, the target is
$$
\min_{\thbg}\Ss(\thbg,\wwd),
$$ 
performed by stochastic gradient descent. 
\begin{remark}
In the classical GAN setting, this optimisation problem may fail to converge~\cite{GoodfellowFail}.
Over the past few years, progress has been made to improve the convergence quality of the algorithm and to improve its stability, 
using different loss functions or adding regularising terms.
We refer the interested reader to the corresponding papers~\cite{GANImprov4, GANImprov6, GANImprov, GANImprov3, GANImprov2, GANImprov5, GANImprov7},
and leave it to future research to integrate these improvements into a quantum setting.
\end{remark}
\begin{proposition}
The solution $(\thbgs, \ww_{D}^*)$ to the $\min-\max$ problem~\eqref{eq:MinMax} is such that the wave function 
$\ket{v_{\thbgs}}$ satisfies $|\bra{\psi_{\tar}}\ket{v_{\thbgs}}|^2=1$, namely, for each
$i\in\{0,\ldots,2^n-1\}$, 
$$
\PP(\ket{\psi_{\tar}})=\ket{i})=\PP(\ket{v_{\thbgs}}=\ket{i}).
$$
\end{proposition}
\begin{proof}
Define the density matrices $\rho_{\tar}:=\ket{\psi_{\tar}}\bra{\psi_{\tar}}$ and 
$\rho_{\thbg}:=\ket{v_{\thbg}}\bra{v_{\thbg}}$
as well as the operator
$P_{\wwd}^R := \Dr(\wwd)^{\dagger}\Pi_1^{\dagger}\Pi_1\Dr(\wwd)$.
Then
$$
\Ss(\thbg,\wwd)= \Tr\left(P_{\wwd}^R\{\rho_{\tar}-\rho_{\thbg}\}\right)
$$
Since $\Pi_1+\Pi_0=\Id$ and $\Dr(\wwd)$ is unitary, setting 
$P_{\wwd}^F := \Dr(\wwd)^{\dagger}\Pi_0^{\dagger}\Pi_0\Dr(\wwd)$, 
it is straightforward to rewrite $\Ss(\thbg,\wwd)$ as
$$
\Ss(\thbg,\wwd)
 = \Tr\left(P_{\wwd}^R\rho_{\tar})+\Tr(P_{\wwd}^F\rho_{\thbg}\right) - 1,
$$
since $\Tr(\rho_{\thbg})=1$ according to the Born Rule (Theorem~\ref{thm:Born}) 
and $P_{\wwd}^R+P_{\wwd}^F=\Id$. 
Again, we also have
$$
\Ss(\thbg,\wwd)
 = -1 + \frac{1}{2}\Tr\left(
 \left(P_{\wwd}^R+P_{\wwd}^F\right)
 \left(\rho_{\tar}+\rho_{\thbg}\right)\right)
 + \frac{1}{2}\Tr\left(
 \left(P_{\wwd}^R-P_{\wwd}^F\right)
 \left(\rho_{\tar}-\rho_{\thbg}\right)\right),
$$
and finally
$$
\Ss(\thbg,\wwd)
 = \frac{1}{2}\Tr\left(\left(P_{\wwd}^R-P_{\wwd}^F\right)\left(\rho_{\tar}-\rho_{\thbg}\right)\right).
$$
Recall that for two Hermitian matrices $A, B$, the inequality
$\Tr(AB)\leq \|A\|_p \|B\|_q$ holds
for $p,q\geq 1$ with $\frac{1}{p}+\frac{1}{q}=1$,
where $\|\cdot\|_p$ denotes the $p$-norm.
Since $P_{\wwd}^R$ and $P_{\wwd}^F$ are Hermitian, 
we obtain (with $p=\infty$ and $q=1$)
$$
\Ss(\thbg,\wwd)\leq\frac{1}{2}
\left\|P_{\wwd}^R-P_{\wwd}^F\right\|_{\infty}
\left\|\rho_{\tar}-\rho_{\thbg}\right\|_1,
$$
where $\left\|P_{\wwd}^R-P_{\wwd}^F\right\|_{\infty}\leq1$.
Thus the optimal $\wwd^*$ satisfies
$$
\max_{\wwd}\Ss(\thbg,\wwd)=\Ss(\thbg,\wwd^*)=\frac{1}{2}\left\|\rho_{\tar}-\rho_{\thbg}\right\|_1.
$$
Again, since $\|\rho_{\tar}-\rho_{\thbg}\|_1\geq0$ the optimal $\thbg^*$ gives
$$
\min_{\thbg}\max_{\wwd}\Ss(\thbg,\wwd)=\Ss(\thbg^*,\wwd^*)=0,
$$
which is equivalent to 
$\|\rho_{\tar}-\rho_{\thbg}\|_1=0$,
itself also equivalent to
$\PP(\ket{v_{\thbg^*}}=\ket{i})=\PP(\ket{\psi_{\tar}}=\ket{i})=p_i$, 
for all $i \in \{0,\ldots,2^n-1\}$.
\end{proof}

\begin{remark}
Our strategy to reach and approximate a solution to the $\min-\max$ problem will be as follows:
we train the discriminator by stochastic gradient ascent~$n_D$ times and then train the generator~$n_G$ times 
by stochastic gradient descent and repeat this~$\ek$ times.
\end{remark}

\section{Financial application: SVI goes Quantum}\label{sec:FinAppli}
We provide here a simple example of generating data in a financial context with the aim to increase interdisciplinarity between quantitative finance and quantum computing.

\subsection{Financial background and motivation}
Some of the most standard and liquid traded financial derivatives are so-called European Call and Put options.
A Call (resp. Put) gives its holder the right, but not the obligation, to buy (resp. sell) an asset at a specified price (the strike price~$K$) at a given future time (the maturity~$T$).
Mathematically, the setup is that of a filtered probability space $(\Omega, \Ff,(\Ff_t)_{t\geq0}, \PP)$
where~$(\Ff_t)_{t\geq 0}$ represents the flow of information; on this space, an asset~$S=(S_t)_{t\geq0}$ is traded and assumed to be adapted (namely $S_t$ is $\Ff_t$-measurable for each $t\geq 0$).
We further assume that there exists a probability~$\QQ$, equivalent to~$\PP$ such that $S$ is a $\QQ$-martingale.
This martingale assumption is key as the Fundamental Theorem of Asset Pricing~\cite{FTAP} in particular implies that this is equivalent to Call and Put prices being respectively equal, at inception of the contract, to
$$
\Cr(K,T) = \EE[\max(S_T-K, 0)|\Ff_0]
\qquad\text{and}\qquad
\Put(K,T) = \EE[\max(K-S_T, 0)|\Ff_0],
$$
where the expectation~$\EE$ is taken under the risk-neutral probability~$\QQ$.
Under sufficient smoothness property of the law of~$S_T$, differentiating twice the Call price yields that the probability density function of the log stock price $\log(S_T)$ is given by
\begin{equation}\label{proba_log_stock}
    p_T(k) = \left(\frac{\partial^2\Cr(K,T)}{\partial K^2}\right)_{K=S_0\E^{k}},
\end{equation}
implying that the real distribution of the (log) stock price can in principle be recovered from options data.
However, prices are not quoted smoothly in $(K,T)$ and
interpolation and extrapolation are needed.
Doing so at the level or prices turns out to be rather cumbersome and market practice usually does it at the level of the so-called implied volatility.
The basic fundamental model of a continuous-time financial martingale is given by the Black-Scholes model~\cite{Black-Scholes}, under which 
$$
\frac{\D S_t}{S_t} = \sigma \D W_t,
\qquad S_0>0,
$$
where $\sigma>0$ is the (constant) instantaneous volatility and~$W$ a standard Brownian motion adapted to the filtration $(\Ff_t)_{t\geq 0}$.
In this model, Call prices admit the closed-form formula
$$
\Cr_{\BS}(K,T,\sigma) :=\EE[\max(S_T-K, 0)|\mathcal{F}_0] = S_0 \BS\left(\log\left(\frac{K}{S_0}\right), \sigma^2 T\right),
$$
where 
$$
\BS(k,v) := \left\{
    \begin{array}{ll}
        \displaystyle \mathcal{N}(d_+(k,v)) - \E^k\mathcal{N}(d_-(k,v)), & \mbox{if } v>0, \\
        (1-\E^k)_+, & \mbox{if } v=0,
    \end{array}
\right.
$$
with $d_{\pm}(k,v):=-\frac{k}{\sqrt{v}} \pm \frac{\sqrt{v}}{2}$, where $\mathcal{N}$ denotes the cumulative distribution function of the Gaussian distribution. 
With a slight abuse of notation, we shall from now on write
$C_{\BS}(K,T,\sigma)=C_{\BS}(k, T, \sigma)$, 
where $k:= \log(\frac{K}{S_0})$ represents the logmoneyness.

\begin{definition}
Given a strike~$K\geq 0$, a maturity~$T\geq 0$ and a Call price $C(K,T)$ (either quoted on the market or computed from a model), 
the implied volatility $\sigma_{\imp}(k,T)$ is defined as the unique non-negative solution to the equation
\begin{equation}\label{implied vol}
    \Cr_{\BS}(k, T, \sigma_{\imp}(k,T))=\Cr(K,T).
\end{equation}
\end{definition}
Note that this equation may not always admit a solution.
However, under no-arbitrage assumptions (equivalently under bound constraints for $C(K,T)$), it does so.
We refer the interested reader to the volatility bible~\cite{GatheralBook} for full explanations of these subtle details.
It turns out that the implied volatility is a much nicer object to work with (both practically and academically); plugging this definition into~\eqref{proba_log_stock}
yields that the map $k\mapsto \sigma_{\imp}(k,T)$ fully characterises the distribution of~$\log(S_T)$ as
\begin{equation}\label{imp_proba_log_stock}
    p_T(k) 
= \left(\frac{\partial^2 \Cr_{\BS}(k, T, \sigma_{\imp}(k,T))}{\partial K^2}\right)_{K=S_0\E^{k}}.
\end{equation}
While a smooth input $\sigma_{\imp}(\cdot,T))$ is still needed, it is however easier than for option prices.
A market standard is the Stochastic Volatility Inspired (SVI) parameterisation proposed by Gatheral~\cite{SVI} (and improved in~\cite{SSVI,GSSVI}),
where the total implied variance 
$w_{\SVI}(k,T):=\sigma_{\imp}^2(k,T)T$ is assumed to satisfy
\begin{equation}\label{eq:SVI}
w_{\SVI}(k,T) = a+b\left(k-m + \rho\sqrt{(k-m)^2+\xi^2}\right),
\qquad\text{for any }k \in \RR,
\end{equation}
with the parameters $\rho\in [-1,1]$, $a,b,\xi \geq 0$ and $m \in \RR$. 
The probability density function~\eqref{proba_log_stock} of the log stock price then admits the closed-form expression~\cite{SVI}
\begin{equation}\label{eq:closed_form_BS}
p_T(k)
= \frac{g_{\SVI}(k,T)}{\sqrt{2\pi w_{\SVI}(k, T)}}\exp\left\{-\frac{d_{-}(k,w_{\SVI}(k,T))^2}{2}\right\},
\end{equation}
where 
$$
g_{\SVI}(k,T)
 := \left(1-\frac{k w'_{\SVI}(k,T)}{2 w_{\SVI}(k,T)}\right)^2
  - \frac{w'_{\SVI}(k,T)^2}{4}\left(\frac{1}{4}+\frac{1}{w_{\SVI}(k,T)}\right) + \frac{w''_{\SVI}(k,T)}{2},
$$
where all the derivatives are taken with respect to~$k$.
In Figure~\ref{fig:SVIDensity}, we plot the typical shape of the implied volatility smile, together with the corresponding density for the following parameters:
\begin{equation}\label{eq:SVIParams}
a =0.030358 ,\qquad
b = 0.0503815,\qquad
\rho = -0.1 ,\qquad
m =0.3 ,\qquad
\xi = 0.048922 ,\qquad
T = 1.
\end{equation}

\begin{figure}[h!]
    \centering
    \includegraphics[scale=0.5]{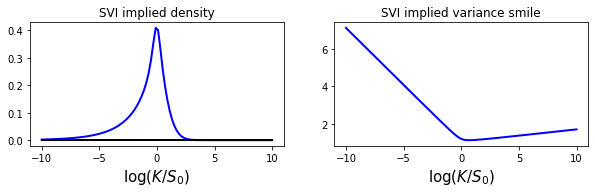}
    \caption{Density of~$\log(S_T)$ computed from~\eqref{eq:closed_form_BS} and the corresponding SVI total variance~\eqref{eq:SVI}. 
    The parameters are given in~\eqref{eq:SVIParams}.}
    \label{fig:SVIDensity}
\end{figure}

\subsection{Numerics}\label{Data of study}
The goal of this numerical part is to be able to generate discrete versions of the SVI probability distribution given in~\eqref{eq:closed_form_BS}.
Our target distribution shall be the one plotted in Figure~\ref{fig:SVIDensity}, corresponding to the parameters~\eqref{eq:SVIParams}.
Since the Quantum GAN (likewise for the classical GAN) algorithm starts from a discrete distribution, 
we first need to discretise the SVI one.
For convenience, we normalise the distribution on the closed interval $[-1,1]$ and discretise with the uniform grid.
$$
\left\{\left\lfloor(2^n-1)\left(\frac{k+1}{2}\right)\right\rfloor\right\}_{k=0,\ldots, 2^n-1},
$$
which we then convert into binary form.
This uniform discretisation does not take into account the SVI probability masses at each point,
and a clear refinement would be to use a one-dimensional quantisation of the SVI distribution.
Indeed, the latter (see~\cite{QuantizationPages} for full details about the methodology) minimises the distance 
(with respect to some chosen norm) between the initial distribution and its discretised version.
We leave this precise study and its error analysis to further research, in the fear that it would clutter the present description of the algorithm.
The discretised distribution, with~$n$ qubits, together with the binary mapping, is plotted in Figure~\ref{fig:distrib} and gives rise to the wave function
$$
\ket{\psi_{\tar}}
 = \sum_{i=0}^{2^n-1}\sqrt{p_i}\ket{i},
 $$
where, for each $i \in \{0,\ldots,2^n-1\}$, 
$$
p_i
 = \PP\left(\log(S_T)\in\Bigg[
 -1+\frac{2i}{2^n},-1+\frac{2(i+1)}{2^n}
 \Bigg)\right).
$$
\begin{figure}
\centering
\includegraphics[scale=0.4]{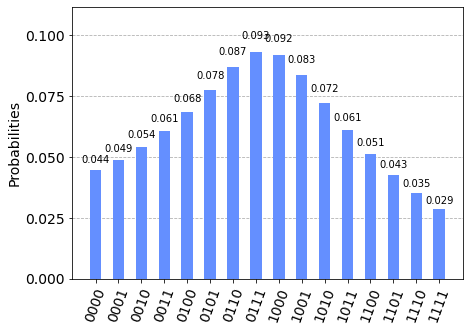}
\caption{Discretised version for the distribution of~$\log(S_T)$ on $[-1,1]$ with $2^4$ points.}
\label{fig:distrib}
\end{figure}

We need metrics to monitor the training of our QuGAN algorithm,
for example the Fidelity function~\cite[Chapter 9.2.2]{QBook}
$$
\Ff: \ket{v_1},\ket{v_2}\in \CC^{2^n}\times \CC^{2^n} \mapsto |\bra{v_1}\ket{v_2}|,
$$
so that for the wave function~\eqref{eq:WaveFunction} $\ket{v_{\thbg}}=\sum_{i=0}^{2^n-1}v_{i,\thbg}\ket{i}$, the goal is to obtain $\Ff\left(\ket{v_{i,\thbg}},\ket{\psi_{\tar}}\right)=1$, 
which gives
$\PP(\ket{v_{\thbg}} = \ket{i}) = \left|v_{i,\thbg}\right|^2=p_i$,
for all $i\in \{0,\ldots,2^n-1\}$.
The Kullback-Leibler Divergence is also a useful monitoring metric, defined as
$$
\KL(\ket{\psi_{\tar}},\ket{v_{\thbg}})
:=\sum_{i=0}^{2^n-1}p_i\log\left(\frac{p_i}{\left|v_{i,\thbg}\right|^2}\right).
$$

\subsubsection{Training and generated distributions}\label{Training and}
In the training of the QuGAN algorithm, 
in each epoch~$\ek$, we train the discriminator $n_D=9$ times and the generator $n_G=1$. 
The results, in Figure~\ref{Training and}, are quite interesting
as the QuGAN manages to overall learn the SVI distribution.
Aside from the limited number of qubits,
the limitations however could be explained via the expressivity of our network which is only parameterised via $(\theta_i)_{i\in\{1,\ldots,9\}}$ and $(w_i)_{i\in\{1,\ldots,4\}}$ which is clearly not enough. This lack of expressivity is a choice, 
and more parameters deepen the network,
but can create a barren plateau phenomenon~\cite{BarrenPlateau}, 
where the gradient vanishes in $\mathcal{O}(2^{-d})$ where~$d$ is the depth of the network. 
This would in turn require an exponentially larger number of shots to obtain a good enough estimation of~\eqref{partialderiv1},
thereby creating a trade-off between expressivity and trainability in a differentiable manner.

\begin{figure}[h!]\label{Final Results2}
    \centering
    \subfloat[Fidelity during QuGAN training.]{\includegraphics[scale=0.25]{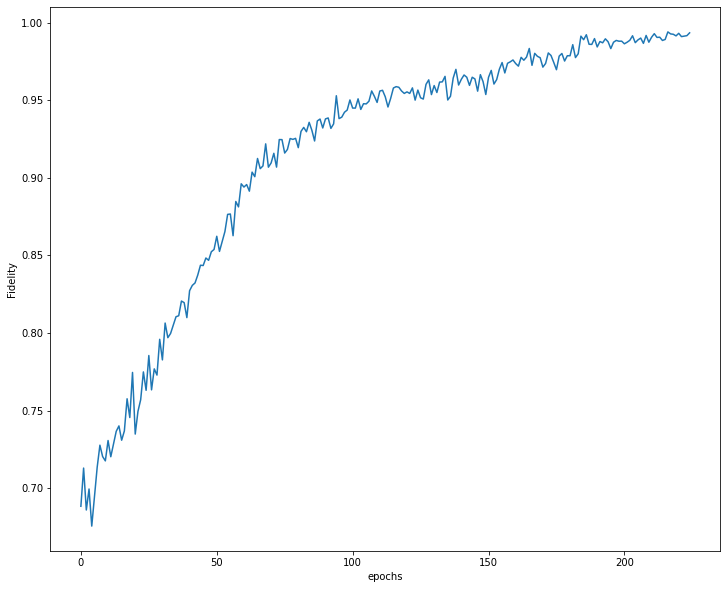}}
    \qquad
    \subfloat[KL Divergence while training the QuGAN.]{\includegraphics[scale=0.25]{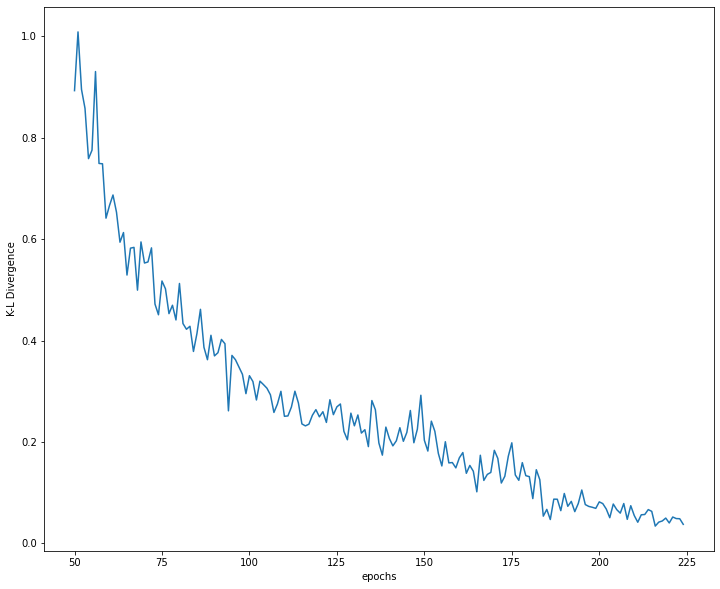} }
\end{figure}

\begin{figure}[h!]\label{Final Results3}
    \subfloat[Evolution of $\|\ket{v_{\thbg}}\bra{v_{\thbg}}-\ket{\psi_{\tar}}\bra{\psi_{\tar}}\|_1$ during QuGAN training. ]{\includegraphics[scale=0.25]{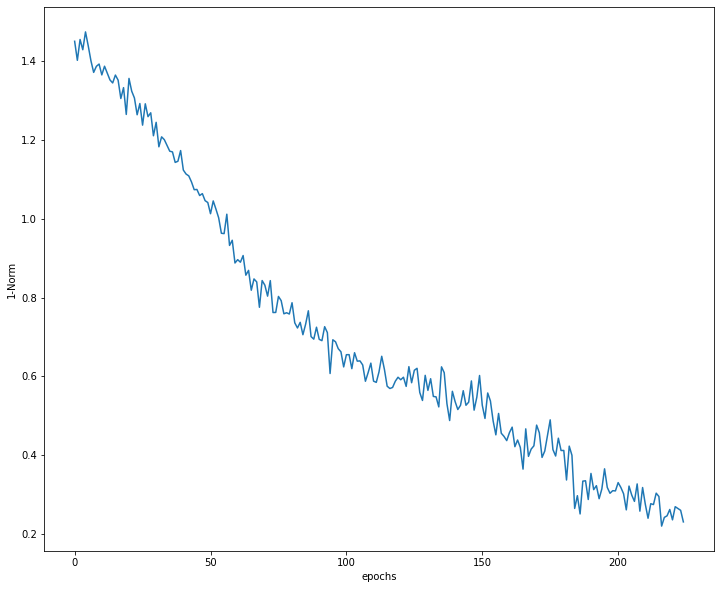} }
    \qquad
    \subfloat[Score function during QuGAN training.]{\includegraphics[scale=0.25]{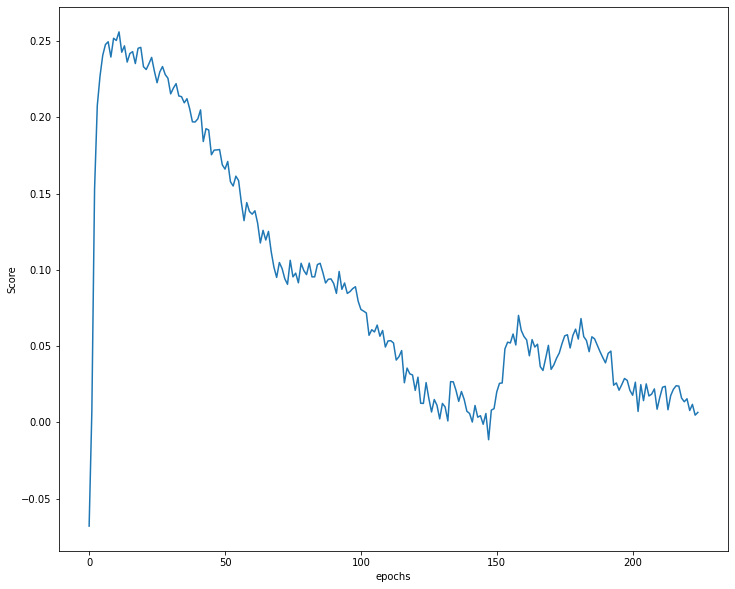}}
\end{figure}

\begin{figure}[h!]\label{Final Results4}
    \subfloat[Generated distribution at $\ek=0$, with uniform random initialisation of  $(\theta)_{i\in\{1,\ldots,9\}}$. ]{\includegraphics[scale=0.4]{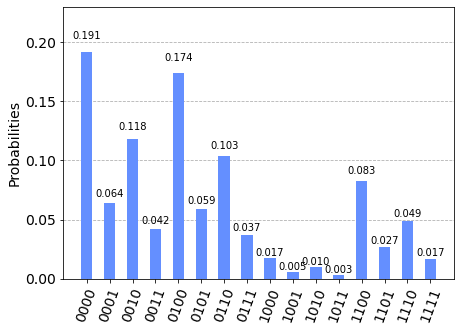}}
    \qquad
    \subfloat[Comparison between the target and the generated distributions at the end of the training.]{\includegraphics[scale=0.4]{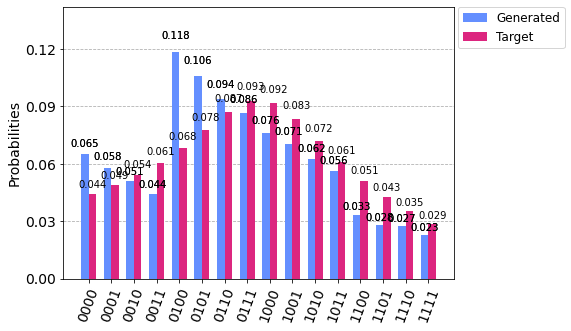}}
\end{figure}

All the numerics in the paper were performed using the \texttt{IBM-Qiskit} library in \texttt{Python}.

\bibliographystyle{siam}
\bibliography{ref}

\begin{thebibliography}{10}

\bibitem{GANImprov4}
{\sc M.~Arjovsky, S.~Chintala, and L.~Bottou}, {\em {W}asserstein {GAN}s}.
\newblock \url{https://arxiv.org/abs/1701.07875}, 2017.

\bibitem{Arnold_1957}
{\sc V.~Arnold}, {\em On functions of three variables}, in Proceedings of the
  {USSR} {A}cademy of {S}ciences, vol.~114, 1957.

\bibitem{NISQ}
{\sc K.~Bharti, A.~Cervera-Lierta, T.~H. Kyaw, T.~Haug, S.~Alperin-Lea,
  A.~Anand, M.~Degroote, H.~Heimonen, J.~S. Kottmann, T.~Menke, W.-K. Mok,
  S.~Sim, L.-C. Kwek, and A.~Aspuru-Guzik}, {\em Noisy intermediate-scale
  quantum ({NISQ}) algorithms}.
\newblock 2021.

\bibitem{Black-Scholes}
{\sc F.~Black and M.~Scholes}, {\em The pricing of options and corporate
  liabilities}, Journal of {P}olitical {E}conomy, 81 (1973), pp.~637--654.

\bibitem{HowToEnhance}
{\sc P.~Braccia, F.~Caruso, and L.~Banchi}, {\em How to enhance quantum
  generative adversarial learning of noisy information}, New Journal of
  Physics, 23 (2021), p.~053024.

\bibitem{FinGAN2}
{\sc H.~Buehler, L.~Gonon, J.~Teichmann, and B.~Wood}, {\em Deep hedging},
  Quantitative Finance, 19 (2019), pp.~1271--1291.

\bibitem{chakrabarti2019quantum}
{\sc S.~Chakrabarti, Y.~Huang, T.~Li, S.~Feizi, and X.~Wu}, {\em Quantum
  wasserstein generative adversarial networks}, 2019.

\bibitem{QuGANDallaire}
{\sc P.-L. Dallaire-Demers and N.~Killoran}, {\em Quantum generative
  adversarial networks}, Physical Review A, 98 (2018).

\bibitem{FTAP}
{\sc F.~Delbaen and W.~Schachermayer}, {\em A general version of the
  fundamental theorem of asset pricing}, Mathematische Annalen, 300 (1994),
  p.~463–520.

\bibitem{GANImprov6}
{\sc E.~Denton, S.~Chintala, A.~Szlam, and R.~Fergus}, {\em Deep generative
  image models using a {L}aplacian pyramid of adversarial networks}, in
  NeurIPS, 2015.

\bibitem{GANImprov}
{\sc I.~Deshpande, Z.~Zhang, and A.~G. Schwing}, {\em Generative modeling using
  the sliced {W}asserstein distance}, in Proceedings of the IEEE conference on
  computer vision and pattern recognition, 2018, pp.~3483--3491.

\bibitem{SVI}
{\sc J.~Gatheral}, {\em A parsimonious arbitrage-free implied volatility
  parameterization with application to the valuation of volatility
  derivatives}.
\newblock 2004.

\bibitem{GatheralBook}
\leavevmode\vrule height 2pt depth -1.6pt width 23pt, {\em The Volatility
  Surface: A Practitioner’s Guide}, Wiley Finance, 2006.

\bibitem{SSVI}
{\sc J.~Gatheral and A.~Jacquier}, {\em Arbitrage-free {SVI} volatility
  surfaces}, Quantitative {F}inance, 14 (2013), pp.~59--71.

\bibitem{goodfellow2014generative}
{\sc I.~Goodfellow, J.~Pouget-Abadie, M.~Mirza, B.~Xu, D.~Warde-Farley,
  S.~Ozair, A.~Courville, and Y.~Bengio}, {\em Generative adversarial nets},
  Advances in neural information processing systems, 27 (2014).

\bibitem{GoodfellowFail}
{\sc I.~J. Goodfellow}, {\em On distinguishability criteria for estimating
  generative models}.
\newblock \url{https://arxiv.org/abs/1412.6515}, 2014.

\bibitem{GANImprov3}
{\sc I.~Gulrajani, F.~Ahmed, M.~Arjovsky, V.~Dumoulin, and A.~C. Courville},
  {\em Improved training of {W}asserstein {GAN}s}, in NeurIPS, 2017,
  p.~5767–5777.

\bibitem{GSSVI}
{\sc G.~Guo, A.~Jacquier, and C.~Martini}, {\em Generalised arbitrage-free
  {SVI} volatility surfaces}, SIAM Journal on Financial Mathematics, 7 (2016),
  pp.~619--641.

\bibitem{QuGANSuper}
{\sc L.~Hu, S.-H. Wu, W.~Cai, Y.~Ma, X.~Mu, Y.~Xu, H.~Wang, Y.~Song, D.-L.
  Deng, C.-L. Zou, et~al.}, {\em Quantum generative adversarial learning in a
  superconducting quantum circuit}, Science {A}dvances, 5 (2019), pp.~27--61.

\bibitem{advQuGAN}
{\sc H.-L. Huang, Y.~Du, M.~Gong, Y.~Zhao, Y.~Wu, C.~Wang, S.~Li, F.~Liang,
  J.~Lin, Y.~Xu, R.~Yang, T.~Liu, M.-H. Hsieh, H.~Deng, H.~Rong, C.-Z. Peng,
  C.-Y. Lu, Y.-A. Chen, D.~Tao, X.~Zhu, and J.-W. Pan}, {\em Experimental
  quantum generative adversarial networks for image generation}, 2021.

\bibitem{GameTheory}
{\sc S.~Kakutani}, {\em A generalization of {B}rouwer’s fixed point theorem},
  Duke Mathematical Journal, 8 (1941), pp.~457--459.

\bibitem{Kolmogorov_1956}
{\sc A.~Kolmogorov}, {\em On the representation of continuous functions of
  several variables by superpositions of continuous functions of a smaller
  number of variables}, in Proceedings of the {USSR} {A}cademy of {S}ciences,
  vol.~108, 1956.

\bibitem{FinGAN1}
{\sc A.~Koshiyama, N.~Firoozye, and P.~Treleaven}, {\em Generative adversarial
  networks for financial trading strategies fine-tuning and combination},
  Quantitative Finance, 21 (2021), pp.~797--813.

\bibitem{QuGAN}
{\sc S.~Lloyd and C.~Weedbrook}, {\em Quantum generative adversarial learning},
  Physical Review Letters, 121 (2018).

\bibitem{NeuroML}
{\sc A.~H. Marblestone, G.~Wayne, and K.~P. Kording}, {\em Toward an
  integration of deep learning and neuroscience}, Frontiers in {C}omputational
  {N}euroscience, 10 (2016), p.~94.

\bibitem{BarrenPlateau}
{\sc J.~R. McClean, S.~Boixo, V.~N. Smelyanskiy, R.~Babbush, and H.~Neven},
  {\em Barren plateaus in quantum neural network training landscapes}, Nature
  {C}ommunications, 9 (2018).

\bibitem{GANImprov2}
{\sc T.~Miyato, T.~Kataoka, M.~Koyama, and Y.~Yoshida}, {\em Spectral
  normalization for {GAN}s}.
\newblock \url{https://arxiv.org/abs/1802.05957}, 2018.

\bibitem{Med1}
{\sc A.~N and H.~P}, {\em Generative modeling for protein structures}, ICLR
  2018 Workshop,  (2018).

\bibitem{FinGAN3}
{\sc H.~Ni, L.~Szpruch, M.~Wiese, S.~Liao, and B.~Xiao}, {\em Conditional
  {S}ig-{W}asserstein {GAN}s for time series generation}.
\newblock \url{https://arxiv.org/abs/2006.05421}, 2020.

\bibitem{QBook}
{\sc M.~Nielsen and I.~Chuang}, {\em Quantum Computation and Quantum
  Information}, CUP, 10~ed., 2010.

\bibitem{nielsen00}
{\sc M.~A. Nielsen and I.~L. Chuang}, {\em Quantum Computation and Quantum
  Information}, Cambridge University Press, 2000.

\bibitem{QENtang}
{\sc M.~Niu, A.~Zlokapa, M.~Broughton, S.~Boixo, M.~Mohseni, V.~Smelyanskyi,
  and H.~Neven}, {\em Entangling quantum {GAN}s}.
\newblock \url{https://arxiv.org/abs/2105.00080}, 2021.

\bibitem{QuantizationPages}
{\sc G.~Pag{\`e}s, H.~Pham, and J.~Printems}, {\em Optimal quantization methods
  and applications to numerical problems in {F}inance}, in Handbook of
  Computational and Numerical Methods in {F}inance, Springer Verlag, 2004.

\bibitem{GANImprov5}
{\sc A.~Radford, L.~Metz, and S.~Chintala}, {\em Unsupervised representation
  learning with deep convolutional generative adversarial networks},  (2016).

\bibitem{Ruf2020}
{\sc J.~Ruf and W.~Wang}, {\em Neural networks for option pricing and hedging:
  a literature review}, Journal of Computational Finance, 24 (2021).

\bibitem{GANImprov7}
{\sc T.~Salimans, I.~Goodfellow, W.~Zaremba, V.~Cheung, A.~Radford, and
  X.~Chen}, {\em Improved techniques for training {GAN}s}, in NeurIPS, 2016.

\bibitem{modecollapse}
{\sc D.~Saxena and J.~Cao}, {\em Generative adversarial networks ({GAN}s):
  Challenges, solutions, and future directions}.
\newblock \url{https://arxiv.org/abs/2005.00065}, 2020.

\bibitem{GenImages2}
{\sc K.~Schawinski, C.~Zhang, H.~Zhang, L.~Fowler, and G.~K. Santhanam}, {\em
  Generative adversarial networks recover features in astrophysical images of
  galaxies beyond the deconvolution limit}, Monthly {N}otices of the {R}oyal
  {A}stronomical {S}ociety: {L}etters, 467 (2017).

\bibitem{ShiftRule}
{\sc M.~Schuld, V.~Bergholm, C.~Gogolin, J.~Izaac, and N.~Killoran}, {\em
  Evaluating analytic gradients on quantum hardware}, Physical Review A, 99
  (2019).

\bibitem{QuGAN2}
{\sc S.~A. Stein, B.~Baheri, R.~M. Tischio, Y.~Mao, Q.~Guan, A.~Li, B.~Fang,
  and S.~Xu}, {\em Qu{GAN}: A generative adversarial network through quantum
  states}.
\newblock 2020.

\bibitem{QuantizedGAN}
{\sc P.~Wang, D.~Wang, Y.~Ji, X.~Xie, H.~Song, X.~Liu, Y.~Lyu, and Y.~Xie},
  {\em {QGAN}: Quantized generative adversarial networks}.
\newblock 2019.

\bibitem{QuantGAN}
{\sc M.~Wiese, R.~Knobloch, R.~Korn, and P.~Kretschmer}, {\em Quant gan: deep
  generation of financial time series}, Quantitative Finance, 20 (2020),
  pp.~1419--1440.

\bibitem{GenImages}
{\sc J.~Yu, Z.~Lin, J.~Yang, X.~Shen, X.~Lu, and T.~Huang}, {\em Generative
  image inpainting with contextual attention}, in Proceedings of the IEEE
  conference on computer vision and pattern recognition, 2018.

\bibitem{Med2}
{\sc A.~Zhavoronkov}, {\em Deep learning enables rapid identification of potent
  {DDR}1 kinase inhibitors}, Nature Biotechnology, 37 (2019), p.~1038–1040.

\end{thebibliography}

\appendix
\section{Review of Quantum Computing techniques and algorithms}\label{sec:Review}
{\small
In Quantum mechanics the state of a physical system is represented by a ket vector~$\ket{v}$ of a Hilbert space~$\Hh$, often $\Hh=\CC^{2^n}$.
Therefore, for a basis $(\ket{0}, \ldots, \ket{2^n-1})$ of~$\Hh$, we obtain  the wave function 
$\ket{v}=\sum_{j=0}^{2^n-1}v_j\ket{j}$.
The Hilbert space is endowed with the inner product $\braket{v|w}$ between two states~$\ket{v}$ and~$\ket{w}$, where $\bra{v}:=\ket{v}^{\dagger}$ is the conjugate transpose.
Recall that a pure quantum state is described by a single ket vector, 
whereas a mixed quantum state cannot.
The following are standard in Quantum Computing,
and we recall them simply to make the paper self-contained.
Full details about these concepts can be found in the excellent monograph~\cite{QBook}.

\begin{theorem}[Born's rule]\label{thm:Born}
If $\ket{v} \in \CC^{2^n}$ be a pure state, then $\|v\|=1$.
\end{theorem}

Given a pure state $\ket{v}=\sum_{j=0}^{2^n-1}v_j\ket{j}$, the probability of measuring $\ket{v}$ collapsing onto the state $\ket{j}$ for $j \in \{0,\ldots,2^n-1\}$ is defined via
\begin{equation} \label{eq:proba}
    \PP(\ket{v}=\ket{j})=|\braket{j|v}|^2
     = \Tr\left(\ket{j}\bra{j} \ket{v}\bra{v}\right)
    =|v_j|^2,
\end{equation}
where $\Tr$ is the Trace operator.
Moreover, for a given state~$\ket{v}$, its density matrix is defined as $\rho_v:=\ket{v}\bra{v}$.

\subsection{Quantum Fourier transform}
In the classical setting, 
the discrete Fourier transform maps a vector $(x_0,\ldots,x_{2^n-1})\in\CC^{2^n}$ to
\begin{equation}\label{fourier}
y_k=\frac{1}{\sqrt{2^n}}\sum_{j=0}^{2^n-1}
\exp\left\{\frac{2\I \pi jk}{2^n}\right\}x_j, 
\qquad\text{for }k =0,\ldots, 2^n-1.
\end{equation}
Similarly, the quantum Fourier transform is the linear operator 
\begin{equation}\label{QFT}
    \ket{j}\mapsto\frac{1}{\sqrt{2^n}}\sum_{k=0}^{2^n-1}\exp\left\{\frac{2\I \pi jk}{2^n}\right\}\ket{k},
\end{equation}
and the operator
$$
\Ffq := \frac{1}{\sqrt{2^n}}\sum_{j,k=0}^{2^n-1}\exp\left\{\frac{2\I \pi jk}{2^n}\right\}\ket{k}\bra{j}
$$ 
represents the Fourier transform matrix which is unitary as $\Ffq*\Ffq^{\dagger}=\Id$.
In an  $n$-qubit system ($\Hh=\CC^{2^n}$) 
with basis $(\ket{0},\ldots,\ket{2^n-1})$; 
for a given state~$\ket{j}$, we use the binary representation 
\begin{equation}\label{eq:binRepr}
j := \overline{j_1\cdots j_n},
\end{equation}
with $(j_1, \ldots, j_n)\in\{0,1\}^n$ so that
$\ket{j}=\ket{j_1\cdots j_n}=\ket{j_1}\otimes\ldots\otimes \ket{j_n}$.
Likewise, the notation $0.j_1j_2 \ldots j_n$ represents
the binary fraction $\sum_{i=1}^{n}2^{-i} j_{i}$.
Elementary algebra then yields
\begin{equation}\label{final fourier}
    \Ffq \ket{j}
    = \frac{1}{2^{\frac{n}{2}}}
    \left(\ket{0}+\E^{2\I\pi 0.j_n }\ket{1}\right)
    \otimes
    \left(\ket{0}+\E^{2\I\pi 0.j_{n-1}j_n }\ket{1}\right)
    \otimes\cdots\otimes
    \left(\ket{0}+\E^{2\I\pi 0.j_1\ldots j_n} \ket{1}\right).
\end{equation}

\subsection{Quantum phase estimation (QPE)}\label{sec:QPE}
The goal of QPE is to estimate the unknown phase $\phi \in [0,1)$
for a given unitary operator~$\Ug$ with an eigenvector~$\ket{u}$ and eigenvalue $\E^{2\I\pi \phi}$. 
Consider a register of size~$\qt$, so that $\Hh=\CC^{2^{\qt}}$ and define
$
b^* := \sup_{j\leq2^{\qt}\phi}\left\{j = 2^{\qt} 0.j_1\cdots j_{\qt}\right\}$.
Thus with $b^*=\overline{b_1\cdots b_{\qt}}$, 
we obtain that $2^{-\qt} b^*=0.b_1\cdots b_{\qt}$ 
is the best $\qt$-bit approximation of~$\phi$ from below.
The quantum phase estimation procedure uses two registers. 
The first register contains the~$\qt$ qubits initially in the state~$\ket{0}$. 
Selecting~$\qt$ relies on 
the number of digits of accuracy for the estimate for~$\phi$, and the probability for which we wish to obtain a successful phase estimation procedure. 
Up to a \textrm{SWAP} transformation, the quantum phase circuit gives the output
$$
\ket{\psi_{\out}}=
\frac{\left(\ket{0}+\E^{2\I\pi 0.\phi_m}\ket{1}\right)
\otimes
\left(\ket{0}+\E^{2\I\pi 0.\phi_{m-1}\phi_m}\ket{1}\right)
\otimes
\cdots
\otimes
\left(\ket{0}+\E^{2\I\pi 0.\phi_1\ldots\phi_m} \ket{1}\right)
}{2^{\frac{m}{2}}},
$$
which is exactly equal to the Quantum Fourier Transform for the state $\ket{2^m\phi}=\ket{\phi_1\phi_2\ldots\phi_m}$ as in~\eqref{final fourier}, and therefore
$\ket{\psi_{state}}=\Ffq\ket{2^m\phi}$.
Since the Quantum Fourier Transform is a unitary transformation, 
we can inverse the process to retrieve $\ket{2^m\phi}$.
Algorithm~\ref{algo:QPE} below provides pseudo-code for the Quantum Phase Estimation procedure 
and we refer the interested reader to~\cite[Chapter 5.2]{QBook}
for detailed explanations.

\begin{algorithm}
\begin{algorithmic}
      \State{\textbf{Input:} Unitary matrix~$\Ug$ with $\Ug\ket{u} = \E^{2\I\pi \phi}\ket{u}$; $\qt=n+ \lceil \log(2+\frac{1}{2\eps})\rceil$ ancilla qubits initialised at $\ket{0}$.} 
     \newline
      \State{\textbf{Procedure:}} 
      \newline
    \State \textbf{1}.  $\ket{0}^{\otimes \qt}\ket{u}$ \Comment{Initial state with $\ket{0}^{\otimes \qt}$ being the ancilla register and $\ket{u}$ the eigenstate register} 
    \newline
    \State \textbf{2}. $\displaystyle\longrightarrow \frac{1}{\sqrt{2^\qt}}\sum_{j=0}^{2^\qt-1}\ket{j}\ket{u}$ \Comment{Hadamard gates applied to the ancilla register}
    \newline
    \State \textbf{3}. $\displaystyle\longrightarrow \frac{1}{\sqrt{2^\qt}}\sum_{j=0}^{2^\qt-1}\ket{j}\Ug^j\ket{u} = \frac{1}{\sqrt{2^\qt}}\sum_{j=0}^{2^\qt-1}\ket{j}\E^{2\I\pi j\phi}\ket{u}$ \Comment{Controlled~$\Ug^j$ gates applied to the eigenstate register
    }
    \newline
    \State \textbf{4}.$ \xrightarrow[]{} \ket{\widetilde{\phi}}\ket{u}$ \Comment{Apply the inverse QFT where $\widetilde{\phi}$ is a $\qt$-qubit approximation of~$\phi$ with an accuracy of $2^{-n}$}
    \newline\State \textbf{5}.$ \xrightarrow[]{} \widetilde{\phi}$ \Comment{Measure $\widetilde{\phi}$ with a probability at least $1-\eps$       }
\newline
    \Return \textbf{Output:} $\widetilde{\phi}$
\end{algorithmic}
\caption{$\text{Quantum Phase Estimation}~(\Ug,\ket{u},\qt, \eps)$}
\label{algo:QPE}
\end{algorithm}
}
\end{document}